\pdfoutput=1
\documentclass[12pt,fleqn]{article}
\usepackage[utf8]{inputenc}
\usepackage[T1]{fontenc}

\usepackage{comment}
\usepackage{fullpage}

\usepackage{amsmath,amsthm,amssymb}
\usepackage{mathtools,thm-restate}
\usepackage{authblk}

\usepackage{fouriernc}
\usepackage[cal=pxtx,scr=rsfs,bb=txof,frak=pxtx]{mathalpha}
\let\newmathbb\mathbb
\AtBeginDocument{%
    \let\mathbb\relax\newcommand{\mathbb}[1]{\bm{\newmathbb{#1}}}
}
\usepackage{esvect} 

\usepackage{enumitem}
\setlist{itemsep=1mm}

\newcommand{\uppsi}{\otherpsi}

\usepackage[style=alphabetic,natbib=true,sortcites=true,backref=true,maxnames=99,maxalphanames=99]{biblatex}

\DefineBibliographyStrings{english}{%
  backrefpage = {$\Lsh$~p.},
  backrefpages = {$\Lsh$~pp.},
}
\usepackage[colorlinks=true,citecolor=blue!65!black,linkcolor=red!70!black]{hyperref}
\usepackage[nameinlink]{cleveref}

\usepackage{subfig}
\usepackage{booktabs,colortbl,multirow}
\usepackage{diagbox,enumitem}
\usepackage{ascmac}
\usepackage[normalem]{ulem}

\usepackage[thinlines]{easytable}  
\usepackage{pifont}  
\usepackage{wrapfig}  

\usepackage{comment}
\usepackage{url}

\usepackage[colorinlistoftodos,prependcaption]{todonotes}

\usepackage{xspace}
\usepackage{bm}

\usepackage{tikz}
\usepackage{tikz-3dplot}
\usetikzlibrary{math}
\usetikzlibrary{positioning}
\usetikzlibrary{decorations.markings}
\usetikzlibrary{shapes.geometric}
\usetikzlibrary{arrows.meta}
\usetikzlibrary{arrows,automata,decorations.pathmorphing,fit,positioning,arrows.meta,calc}

\crefname{theorem}{Theorem}{Theorems}
\crefname{proposition}{Proposition}{Propositions}
\crefname{lemma}{Lemma}{Lemmas}
\crefname{claim}{Claim}{Claims}
\crefname{corollary}{Corollary}{Corollaries}
\crefname{observation}{Observation}{Observations}
\crefname{remark}{Remark}{Remarks}
\crefname{example}{Example}{Examples}
\crefname{hypothesis}{Hypothesis}{Hypotheses}
\crefname{definition}{Definition}{Definitions}
\crefname{problem}{Problem}{Problems}

\crefname{section}{Section}{Sections}
\crefname{appendix}{Appendix}{Appendices}
\crefname{equation}{Eq.}{Eqs.}
\crefname{table}{Table}{Tables}
\crefname{figure}{Figure}{Figures}

\newcommand{\hl}[1]{\textcolor{red}{\uwave{\textcolor{black}{#1}}}}

\renewcommand{\geq}{\geqslant}
\renewcommand{\leq}{\leqslant}
\renewcommand{\epsilon}{\varepsilon}
\renewcommand{\rho}{\varrho}
\renewcommand{\phi}{\varphi}
\renewcommand{\bar}{\overline}
\renewcommand{\hat}{\widehat}

\newcommand{\prb}[1]{\textup{\textsc{#1}}\xspace}
\renewcommand{\prb}[1]{\textup{\textsf{#1}}\xspace}
\newcommand{\nth}[1]{#1\textsuperscript{th}\xspace}

\renewcommand{\vec}[1]{\mathbf{\bm{#1}}}
\newcommand{\mat}[1]{\mathbf{\bm{#1}}}

\newcommand{\reco}{\leftrightsquigarrow}
\newcommand{\rme}{\mathrm{e}}

\DeclareMathOperator*{\argmax}{argmax}
\let\Pr\relax\DeclareMathOperator*{\Pr}{\mathbb{Pr}}
\DeclareMathOperator{\E}{\mathbb{E}}
\DeclareMathOperator{\bigO}{\mathcal{O}}
\DeclareMathOperator{\val}{\mathsf{val}}
\DeclareMathOperator{\cost}{\mathsf{cost}}
\DeclareMathOperator{\PLR}{\mathsf{POP}}
\DeclareMathOperator{\OPT}{\mathsf{opt}}
\DeclareMathOperator{\Enc}{\mathsf{Enc}}

\newcommand{\ini}{\mathsf{ini}}
\newcommand{\tar}{\mathsf{tar}}
\newcommand{\ttt}{t}
\newcommand{\TTT}{T}
\newcommand{\aaa}{\mathtt{a}}
\newcommand{\bbb}{\mathtt{b}}
\newcommand{\ccc}{\mathtt{c}}
\newcommand{\rep}{\mathsf{r}}
\newcommand{\Rep}{\mathsf{R}}
\newcommand{\approxone}{\vartheta}

\newcommand{\sqpsi}{\vec{\uppsi}}

\newcommand{\YES}{\textup{\textsc{yes}}\xspace}
\newcommand{\NO}{\textup{\textsc{no}}\xspace}

\newcommand{\cP}{\textup{\textbf{P}}\xspace}
\newcommand{\NP}{\textup{\textbf{NP}}\xspace}
\newcommand{\PSPACE}{\textup{\textbf{PSPACE}}\xspace}

\newcommand{\BCSP}[1][]{\prb{2-CSP{#1}}}
\newcommand{\BCSPReconf}[1][]{\prb{2-CSP{#1} Reconfiguration}}
\newcommand{\MaxminBCSPReconf}[1][]{\prb{Maxmin 2-CSP{#1} Reconfiguration}}
\newcommand{\QCSPReconf}[1][]{\prb{4-CSP{#1} Reconfiguration}}
\newcommand{\MaxminQCSPReconf}[1][]{\prb{Maxmin 4-CSP{#1} Reconfiguration}}
\newcommand{\SetCoverReconf}{\prb{Set Cover Reconfiguration}}
\newcommand{\MinmaxSetCoverReconf}{\prb{Minmax Set Cover Reconfiguration}}
\newcommand{\DominatingSetReconf}{\prb{Dominating Set Reconfiguration}}
\newcommand{\MinmaxDominatingSetReconf}{\prb{Minmax Dominating Set Reconfiguration}}

\newcommand{\blackred}{red!60!black}
\newcommand{\blackgreen}{green!50!black}
\newcommand{\blackblue}{blue!70!black}
\newcommand{\whitered}{red!50!white}
\newcommand{\whitegreen}{green!60!white}
\newcommand{\whiteblue}{blue!50!white}
\newcommand{\RRR}{\textcolor{\blackred}{\mathtt{r}}}
\newcommand{\GGG}{\textcolor{\blackgreen}{\mathtt{g}}}
\newcommand{\BBB}{\textcolor{\blackblue}{\mathtt{b}}}

\usepackage[ruled]{algorithm}
\usepackage{algpseudocodex}

\algrenewcommand\textproc{\textsl}

\newcommand{\calC}{\mathcal{C}}

\newcommand{\calF}{\mathcal{F}}

\newcommand{\calU}{\mathcal{U}}

\newcommand{\bbN}{\mathbb{N}}

\newcommand{\scrC}{\mathscr{C}}
\newcommand{\scrD}{\mathscr{D}}

\newcommand{\scrS}{\mathscr{S}}

\let\Pr\relax\DeclareMathOperator*{\Pr}{\mathbb{P}}

\usepackage{footnote}
\makesavenoteenv{algorithmic}

\newtheorem{theorem}{Theorem}[section]

\newtheorem{lemma}[theorem]{Lemma}
\newtheorem{claim}[theorem]{Claim}
\newtheorem{corollary}[theorem]{Corollary}
\newtheorem{observation}[theorem]{Observation}
\newtheorem{remark}[theorem]{Remark}

\theoremstyle{definition}
\newtheorem{definition}[theorem]{Definition}
\newtheorem{problem}[theorem]{Problem}

\numberwithin{equation}{section}

\title{Gap Amplification for Reconfiguration Problems\footnote{
A preliminary version of this paper appeared in
\emph{Proc.~35th Annu.~ACM-SIAM Symp.~on Discrete Algorithms} (SODA 2024) \cite{ohsaka2024gap}.
}}
\author{Naoto Ohsaka\thanks{
    CyberAgent, Inc., Tokyo, Japan.
    \href{mailto:ohsaka\_naoto@cyberagent.co.jp}{\texttt{ohsaka\_naoto@cyberagent.co.jp}}; 
    \href{mailto:naoto.ohsaka@gmail.com}{\texttt{naoto.ohsaka@gmail.com}}
}}
\date{\today}

\begin{document}

\maketitle
\thispagestyle{empty}
\begin{abstract}\noindent\emph{Combinatorial reconfiguration} is
a brand-new field in theoretical computer science that studies
the reachability and connectivity over the solution space of a particular combinatorial problem.
We study the hardness of accomplishing
``approximate'' reconfigurability,
which affords to relax the feasibility of intermediate solutions.
For example,
in \prb{Minmax Set Cover Reconfiguration},
given a set system $\mathcal{F}$ over a universe $\mathcal{U}$ and
its two covers $\mathcal{C}^\mathsf{ini}$ and $\mathcal{C}^\mathsf{tar}$,
we are required to transform $\mathcal{C}^\mathsf{ini}$ into $\mathcal{C}^\mathsf{tar}$
by repeatedly adding or removing a single set of $\mathcal{F}$,
while covering $\mathcal{U}$ anytime,
so as to minimize the \emph{maximum size} of covers during transformation.
The recent study by {Ohsaka}~(STACS 2023)~\cite{ohsaka2023gap} gives
evidence that a host of reconfiguration problems
are \PSPACE-hard to approximate assuming the \emph{Reconfiguration Inapproximability Hypothesis} (RIH),
which postulates that a gap version of \prb{Maxmin CSP Reconfiguration} is \PSPACE-hard.
One limitation of this approach is that
inapproximability factors are not explicitly shown, so that
even a $1.00 \cdots 001$-approximation algorithm for \prb{Minmax Set Cover Reconfiguration}
may not be ruled out, whereas it admits a $2$-factor approximation algorithm as per
{Ito, Demaine, Harvey, Papadimitriou, Sideri, Uehara, and Uno}~(Theor.~Comput.~Sci., 2011)~\cite{ito2011complexity}.

In this paper, we demonstrate \emph{gap amplification} for reconfiguration problems.
In particular,
we prove an explicit factor of \PSPACE-hardness of approximation 
for three popular reconfiguration problems only assuming RIH.
Our main result is that under RIH,
\prb{Maxmin 2-CSP Reconfiguration} is \PSPACE-hard to approximate within a factor of $0.9942$.
Moreover, the same result holds even if the constraint graph
is restricted to $(d,\lambda)$-expander for arbitrarily small $\frac{\lambda}{d}$.
The crux of its proof is
an alteration of the gap amplification technique due to {Dinur}~(J.~ACM, 2007)~\cite{dinur2007pcp},
which amplifies the $1$~vs.~$1-\varepsilon$ gap for arbitrarily small $\varepsilon \in (0,1)$
up to the $1$~vs.~$1-0.0058$ gap.
As an application of the main result, we demonstrate that
\prb{Minmax Set Cover Reconfiguration} and
\prb{Minmax Dominating Set Reconfiguration}
are \PSPACE-hard to approximate within a factor of $1.0029$ under RIH.
Our proof is based on
a gap-preserving reduction from \prb{Label Cover} to \prb{Set Cover}
due to {Lund and Yannakakis}~(J.~ACM, 1994)~\cite{lund1994hardness}.
Unlike {Lund--Yannakakis}' reduction,
the expander mixing lemma is essential to use.
We highlight that all results hold unconditionally
as long as ``\PSPACE-hard'' is replaced by ``\NP-hard,'' and
are the first explicit inapproximability results for reconfiguration problems without resorting to the parallel repetition theorem.
We finally complement the main result by showing that
it is \NP-hard to approximate \prb{Maxmin 2-CSP Reconfiguration}
within a factor better than~$\frac{3}{4}$.
\end{abstract}
\tableofcontents

\clearpage

\section{Introduction}
\subsection{Background}
\emph{Combinatorial reconfiguration} is a brand-new field in theoretical computer science that studies
the reachability and connectivity over the solution space of
a particular combinatorial problem, referred to as a \emph{source problem}.
Especially,
we aim to find a ``step-by-step'' transformation between a pair of feasible solutions,
called a \emph{reconfiguration sequence},
while maintaining the feasibility of any intermediate solution.
One of the reconfiguration problems we study in this paper is
\SetCoverReconf~\cite{ito2011complexity}:
given a set system $\calF$ over a universe $\calU$ and
its two covers $\calC^\ini$ and $\calC^\tar$ of size $k$,
we are required to transform $\calC^\ini$ into $\calC^\tar$
by repeatedly adding or removing a single set of $\calF$,
while only passing through covers of size at most $k+1$.
The concept of reconfiguration arises in problems involving transformation and movement,
which may date back to
motion planning \cite{hopcroft1984complexity} and
classical puzzles, including the 15 puzzle \cite{johnson1879notes} and Rubik's Cube.
Since the establishment of the reconfiguration framework due to
\citet*{ito2011complexity},
numerous reconfiguration problems have been defined from
Boolean satisfiability, constraint satisfaction problems, graph problems, and others.

Of particular interest has been to elucidate their computational complexity,
which exhibits the following trend:
\begin{itemize}
\item
    On the intractable side,
    a reconfiguration problem generally becomes \PSPACE-complete if its source problem is intractable (say, \NP-complete);
    e.g., reconfiguration problems of
    \prb{3-SAT} \cite{gopalan2009connectivity},
    \prb{Independent Set} \cite{hearn2005pspace,hearn2009games}, and
    \prb{Set Cover} \cite{ito2011complexity} are all \PSPACE-complete.
\item
    On the tractable side,
    a source problem in \cP frequently leads to a reconfiguration problem that also belongs to \cP, e.g.,
    \prb{Matching} \cite{ito2011complexity} and
    \prb{2-SAT} \cite{gopalan2009connectivity}.
\end{itemize}
Some exceptions are known:
whereas \prb{3-Coloring} is \NP-complete,
its reconfiguration problem is solvable in polynomial time \cite{cereceda2011finding};
\prb{Shortest Path} on a graph is tractable,
but its reconfiguration problem is \PSPACE-complete \cite{bonsma2013complexity}.
Recent studies have also focused on the restricted graph class and the parameterized complexity.
We refer the readers to the surveys by
\citet{heuvel2013complexity,nishimura2018introduction}, and
the Combinatorial Reconfiguration wiki \cite{hoang2024combinatorial}
for more algorithmic and hardness results of reconfiguration problems.

In this paper, we study the hardness of accomplishing ``approximate'' reconfigurability,
which affords to \emph{relax} the feasibility of intermediate solutions.
For example, in \MinmaxSetCoverReconf \cite{ito2011complexity}
--- an \emph{optimization version} of \prb{Set Cover Reconfiguration} --- 
we are enabled to employ any covers of size greater than $k+1$,
but required to minimize the \emph{maximum size} of covers in the reconfiguration sequence.
Solving this optimization problem accurately,
we may find an approximate reconfiguration sequence for \SetCoverReconf, say,
comprising covers whose size is at most $1$\% larger than $k+1$;
in fact, it admits a $2$-factor approximation algorithm as per \citet[Theorem~6]{ito2011complexity}.
Such optimization versions would naturally emerge if
we were informed of the nonexistence of a reconfiguration sequence for the original decision version, or
we already knew a problem of interest to be intractable.

In their seminal work,
\citet[Theorems~4 and~5]{ito2011complexity}
demonstrated \NP-hardness of approximating
optimization versions of \prb{Clique Reconfiguration} and \prb{SAT Reconfiguration}, and
posed \PSPACE-hardness of approximation for reconfiguration problems as an open problem.
Note that their results do not bring \PSPACE-hardness of approximation
because of relying on the \NP-hardness of approximating the corresponding optimization problems
(i.e., \prb{Max Clique}~\cite{hastad1999clique} and \prb{Max SAT}~\cite{hastad2001some}, respectively).
The significance of showing \PSPACE-hardness is that
it is tight because many reconfiguration problems are \PSPACE-complete, and
it disproves the existence of a witness (especially a reconfiguration sequence) of polynomial length under
\NP~$\neq$~\PSPACE.

The recent study by \citet{ohsaka2023gap},
in the spirit of \PSPACE-hardness of approximation,
gives evidence that a host of (optimization versions of) reconfiguration problems are \PSPACE-hard to approximate,
including those of 
\prb{2-CSP},
\prb{3-SAT},
\prb{Independent Set},
\prb{Vertex Cover}, and
\prb{Clique}, as well as
\prb{Nondeterministic Constraint Logic} \cite{hearn2005pspace,hearn2009games}.
Their approach is to design a series of gap-preserving reductions starting from
the \emph{Reconfiguration Inapproximability Hypothesis} (RIH),
a working hypothesis postulating
that a gap version of \prb{Maxmin CSP Reconfiguration} is \PSPACE-hard.\footnote{See \cref{sec:pre} for the precise definition.}
Our study builds on \cite{ohsaka2023gap} and delves into the \emph{degree} of inapproximability.

One limitation of {Ohsaka}'s approach \cite{ohsaka2023gap} is
that inapproximability factors are not explicitly shown:
although \MinmaxSetCoverReconf
is found to be $(1+\epsilon)$-factor inapproximable for some $\epsilon \in (0,1)$ under RIH,
the value of such $\epsilon$ might be so small that
even a $1.00 \cdots 001$-factor approximation algorithm may not be ruled out.
Hence, we seek gap amplifiability of reconfiguration problems.
\emph{Gap amplification} is a (polynomial-time) reduction
that makes a tiny gap for a particular gap problem into a larger gap.
In the \NP regime,
the \emph{parallel repetition theorem} due to \citet{raz1998parallel}
serves as gap amplification, deriving strong (or even optimal) inapproximability results;
e.g., \prb{Set Cover} is \NP-hard to approximate within a factor better than $\ln n$
(where $n$ is the universe size) \cite{lund1994hardness,feige1998threshold,dinur2014analytical}.
At the core of this theorem is to imply
along with the PCP theorem~\cite{arora1998probabilistic,arora1998proof}
that
\prb{Max 2-CSP} is \NP-hard to approximate within any constant factor.
Unfortunately,
for \MaxminBCSPReconf, which is a reconfiguration analogue of \prb{Max 2-CSP},
a naive parallel repetition fails to decrease the soundness error;
it is further known to be approximable within a factor of $\approx \frac{1}{4}$ for some graph classes \cite{ohsaka2023approximate}.
Our research question is thus the following:
\begin{center}
\emph{
Can we derive an explicit factor of \PSPACE-hardness of approximation \\
for reconfiguration problems only assuming RIH by gap amplification?
}
\end{center}

\subsection{Our Results}
We provide a partially positive answer to the above question by demonstrating
gap amplification for reconfiguration problems.
In particular, as summarized in \cref{tab:summary},
\MaxminBCSPReconf,
\MinmaxSetCoverReconf, and
\MinmaxDominatingSetReconf
are \PSPACE-hard to approximate within some universal constant factor under RIH,
even though RIH itself does not specify any value of the gap parameter.
Our reductions start from RIH and preserve the \emph{perfect completeness}; i.e.,
\YES instances have a solution to the decision counterpart.
We highlight that all \PSPACE-hardness results hold unconditionally
as long as ``\PSPACE-hard'' is replaced by ``\NP-hard'' (because so is RIH), and
are the first explicit inapproximability results for reconfiguration problems without resorting to the parallel repetition theorem.

\begin{table}[t]
    \centering
    \small
    \begin{tabular}{c|ll}
    \toprule
    & \multirow{2}{*}{\MaxminBCSPReconf} & \multicolumn{1}{c}{\MinmaxSetCoverReconf and} \\
    &  & \multicolumn{1}{c}{\MinmaxDominatingSetReconf}
    \\
    \midrule
    \PSPACE-hardness &
    \multirow{2}{*}{$\hl{0.9942}$ (\cref{thm:amp})} &
    $1.0029$ (\cref{cor:sc,cor:ds})
    \\
    (under RIH) & &
    $\hl{2-\frac{1}{\operatorname{polyloglog} n}}$ \cite{hirahara2024optimal}$^*$
    \\
    \midrule
    \multirow{4}{*}{\NP-hardness} &
    $0.75+\epsilon$ (\cref{thm:NP-BCSPR}) &
    \multirow{1}{*}{$1.0029$ (\cref{cor:sc,cor:ds})}
    \\
    & $\hl{0.5+\epsilon}$ \cite{karthik2023inapproximability}$^*$ &
    $2-\epsilon$ \cite{karthik2023inapproximability}$^*$
    \\
    & $\frac{143}{144}$ \cite{ito2011complexity,ohsaka2023gap} &
    $\hl{2-\frac{1}{\operatorname{polyloglog} n}}$ \cite{hirahara2024optimal}$^*$
    \\
    \midrule
    \multirow{2}{*}{approximability} &
    $0.25-\bigO\left(\frac{1}{\text{ave.~deg.}}\right)$ \cite{ohsaka2023approximate} &
    \multirow{2}{*}{$\hl{2}$ \cite{ito2011complexity}}
    \\
    & $\hl{0.5-\epsilon}$ \cite{karthik2023inapproximability}$^*$ &
    \\
    \bottomrule
    \end{tabular}
    \caption{
        Summary of the inapproximability results obtained in this paper, where 
        $\epsilon$ indicates any small positive number in $(0,1)$ and
        $n$ indicates the size of the universe or the number of vertices.
        See \cref{sec:pre,sec:sc} for the precise definition of 
        \MaxminBCSPReconf and \MinmaxSetCoverReconf, respectively.
        The best result for each case is \hl{wavy-lined}, and
        references marked with ``$^*$'' represent follow-up work (see \cref{subsec:intro:follow-up}).
    }
    \label{tab:summary}
\end{table}

Our main result is that under RIH,
\MaxminBCSPReconf is \PSPACE-hard to approximate within a factor of $0.9942$ (\cref{thm:amp}).
Succinctly speaking, for a constraint graph $G$ and its two satisfying assignments $\psi^\ini$ and $\psi^\tar$,
\MaxminBCSPReconf requires to
transform $\psi^\ini$ into $\psi^\tar$ by repeatedly changing a single value 
so as to maximize the \emph{minimum fraction} of satisfied edges during the transformation.
Moreover, the same result holds even if
$G$ is restricted to $(d,\lambda)$-expander for arbitrarily small $\frac{\lambda}{d}$,
which is critical for further reducing to \MinmaxSetCoverReconf.
The crux of the proof of \cref{thm:amp} is an alteration of
the \emph{gap amplification} technique due to \citet{dinur2007pcp},
which amplifies the $1$~vs.~$1-\epsilon$ gap for arbitrarily small $\epsilon \in (0,1)$
up to the $1$~vs.~$1-0.0058$ gap:
\begin{lemma}[Gap amplification lemma; informal; see \cref{lem:amp:pow}]
    For every number $\epsilon \in (0,1)$,
    there exists a gap-preserving reduction
    from \MaxminBCSPReconf to itself such that
    \begin{enumerate}
        \item the perfect completeness is preserved, and
        \item if any reconfiguration for the former violates $\epsilon$-fraction of edges, then
        $\approx 0.0058$-fraction of edges must be unsatisfied during any reconfiguration for the latter.
    \end{enumerate}
\end{lemma}\noindent
Note that Dinur's gap amplification for \prb{Max \BCSP} does not amplify the gap beyond $\frac{1}{2}$
\cite{bogdanov2005gap}, which seems to apply to this result.

As an application of the main result,
we demonstrate that \MinmaxSetCoverReconf is \PSPACE-hard to approximate within 
a factor of $1.0029$ under RIH (\cref{cor:sc}).
Since this is the first explicit inapproximability even in terms of \NP-hardness
(to the best of our knowledge) and
\MinmaxSetCoverReconf admits a $2$-factor approximation algorithm \cite{ito2011complexity},
achieving a $1.0029$-factor inapproximability still seems plausible.
Based upon a gap-preserving reduction from \prb{Label Cover} to \prb{Set Cover} due to 
\citet{lund1994hardness},
we develop its reconfiguration analogue:
\begin{theorem}[informal; see \cref{thm:sc}]
    For every number $\epsilon \in (0,1)$,
    there exists a gap-preserving reduction
    from \MaxminBCSPReconf to \MinmaxSetCoverReconf such that
    \begin{enumerate}
        \item the perfect completeness is preserved, and
        \item if any reconfiguration for the former violates $\approx (1-\epsilon^2)$-fraction of edges,
        any reconfiguration for the latter must encounter
        a cover of relative size $2-\epsilon$ to the minimum cover.
    \end{enumerate}
\end{theorem}\noindent
Unlike {Lund--Yannakakis}' reduction,
the expander mixing lemma \cite{alon1988explicit} is essential to use
because the underlying graph of \cref{thm:amp} is not biregular, just expander.
As a corollary,
we show that the same inapproximability result holds for
\MinmaxDominatingSetReconf (\cref{cor:ds} \cite{paz1981non});
an interesting open question is whether this problem admits a $(2-\epsilon)$-factor approximation algorithm.

We finally complement \PSPACE-hardness of approximation for \MaxminBCSPReconf
by showing its \NP-hardness analogue:
it is \NP-hard to approximate \MaxminBCSPReconf
within a factor better than $\frac{3}{4}$,
improving a factor of $1-\frac{1}{144}$ inferred from \cite{ito2011complexity,ohsaka2023gap}.\footnote{
\cite[Theorem 5]{ito2011complexity} implies that
\prb{5-SAT Reconfiguration} is \NP-hard to approximate within a $\left(1-\frac{1}{16}\right)$-factor,
\cite{ohsaka2023gap} thus implies that
\prb{3-SAT Reconfiguration} is \NP-hard to approximate within a $\left(1-\frac{1}{48}\right)$-factor, and then
\cite{ohsaka2023gap} implies that
\BCSPReconf is \NP-hard to approximate within a $\left(1-\frac{1}{144}\right)$-factor.
}
\begin{theorem}[informal; see \cref{thm:NP-BCSPR}]
    For every number $\epsilon \in (0,1)$,
    a $1$ vs.~$\frac{3}{4}+\epsilon$ gap version of
    \MaxminBCSPReconf is \NP-hard.
\end{theorem}\noindent
The proof uses a simple gap-preserving reduction from \prb{Max \BCSP} based on
\cite{fortnow1994power,ito2011complexity},
whose inapproximability relies on the parallel repetition theorem \cite{raz1998parallel}.\footnote{
The underlying graph resulting from
\cref{thm:NP-BCSPR}
is neither biregular nor expander.
}

\subsection{Overview of Gap Amplification of \BCSPReconf}
The proof of \cref{thm:amp} is based on an adaptation of the first two steps of
Dinur's proof \cite{dinur2007pcp} of the PCP theorem \cite{arora1998probabilistic,arora1998proof} (\cref{lem:amp:expand,lem:amp:pow}).
We review each of them followed by our technical challenge.

The first step is \textbf{preprocessing}, which makes a constraint graph into an expander.
This step further comprises two partial steps:
(1) in the \textbf{degree reduction} step
due to \citet{papadimitriou1991optimization},
each vertex is replaced by an expander called a cloud, and 
equality constraints are imposed on the intra-cloud edges so that
the assignments in the cloud behave like a single assignment;
(2) in the \textbf{expanderization} step,
an expander is superimposed on the small-degree graph so that
the resulting graph is entirely expander.
Because a reconfiguration analogue of the \textbf{degree reduction} step was already established~\cite{ohsaka2023gap},
what remains to be done is \textbf{expanderization},
which is easy to implement on \MaxminBCSPReconf (\cref{lem:amp:expand}).

The second step is \textbf{powering}.
Given a $d$-regular expander constraint graph $G$ on vertex set $V$
and alphabet $\Sigma$ and a parameter $\rep \in \bbN$,
this step constructs the \emph{powered constraint graph} $G'$ on the same vertex set as follows
\cite{radhakrishnan2006gap,radhakrishnan2007dinurs}:
The alphabet of $G'$ is extremely enlarged from $G$ so that
    each vertex can have an ``opinion'' about the value of any vertices within distance $\bigO(\rep)$.
Such an assignment of opinions, called a \emph{proof}, is specified by
a function $\psi' \colon V \to \Sigma^{d^{\bigO(\rep)}}$.
Each edge of $G'$ corresponds to a random walk over $G$ whose expected length is $\rep$.
The \emph{verifier} samples such a random walk from $x$ to $y$ and accepts a proof $\psi'$ if
\begin{enumerate}
    \item $x$ and $y$ have the same opinion about the value of each vertex appearing in the walk, and
    \item such opinions satisfy all edges along the walk.
\end{enumerate}
Informally, the verifier becomes $\bigO(\rep)$ times more likely
to reject $\psi'$, whose proof is outlined below \cite{dinur2007pcp,radhakrishnan2006gap,radhakrishnan2007dinurs}:
Given a proof $\psi' \colon V \to \Sigma^{d^{\bigO(\rep)}}$ for $G'$,
we extract from it a representative assignment $\psi \colon V \to \Sigma$ for $G$ such that
$\psi(v)$ is given by the \emph{popularity vote} on opinions about $v$'s value.
Suppose $G$ is $\epsilon$-far from satisfiable; i.e., an $\epsilon$-fraction of its edges must be unsatisfied by $\psi$.
Since the verifier rejects $\psi'$ whenever
the aforementioned verifier's test fails at any of the unsatisfied edges,
it suffices to estimate the probability of encountering such \emph{faulty steps},
which is shown to be $\bigO(\rep \epsilon)$ thanks to expansion properties of $G$.

\begin{table}[t]
\centering%
\begin{TAB}(e,3mm,3mm){c|c:c:c:c:c:c|}{c|c:c:c:c:c:c|}
     & $\aaa$ & $\aaa\bbb$ & $\bbb$ & $\bbb\ccc$ & $\ccc$ & $\ccc\aaa$ \\ 
    $\aaa$ & $\blacksquare$ & $\blacksquare$ & & & & $\blacksquare$ \\  
    $\aaa\bbb$ & $\blacksquare$ & $\blacksquare$ & $\blacksquare$ & & & \\ 
    $\bbb$ & & $\blacksquare$ & $\blacksquare$ & $\blacksquare$ & & \\ 
    $\bbb\ccc$ & & & $\blacksquare$ & $\blacksquare$ & $\blacksquare$ & \\ 
    $\ccc$ & & & & $\blacksquare$ & $\blacksquare$ & $\blacksquare$ \\ 
    $\ccc\aaa$ & $\blacksquare$ & & & & $\blacksquare$ & $\blacksquare$
\end{TAB}
\caption{Consistency over $\Sigma'$.}
\label{tab:consistency}
\end{table}

One might think of simply applying the \textbf{powering} step to \MaxminBCSPReconf.
However, it sacrifices the perfect completeness:
Since many vertices have opinions about a particular vertex's value,
we have to replace such opinions one by one during reconfiguration, giving rise to different opinions
that would be rejected by the verifier. 
This kind of issue also arises in the \textbf{degree reduction} step due to \citet{ohsaka2023gap}.
To bypass this issue,
we apply the \emph{alphabet squaring trick} \cite{ohsaka2023gap},
which modifies the alphabet as if each vertex could take a pair of values \emph{simultaneously}.
For example, if the original alphabet is
$\Sigma = \{\aaa, \bbb, \ccc\}$,
the new alphabet is
$\Sigma' = \{\aaa, \bbb, \ccc, \aaa\bbb, \bbb\ccc, \ccc\aaa\}$.
The entire proof for $G'$ is now specified by a function $\psi' \colon V \to \Sigma'^{d^{\bigO(\rep)}}$.

Our new concept tailored to the verifier's test
is the ``consistency'' over the squared alphabet $\Sigma'$,
alluded in \cite{ohsaka2023gap}:
$\alpha$ and $\beta$ of $\Sigma'$ are said to be \emph{consistent} if 
$\alpha \subseteq \beta$ or $\alpha \supseteq \beta$;
see \cref{tab:consistency}.
Intuitively, this definition of consistency allows us to transform
a group of $\aaa$'s into a group of $\bbb$'s consistently \emph{via} $\aaa\bbb$'s, and thus
we can smoothly redesign the verifier's test and the popularity vote.
Specifically, the \emph{modified verifier} over $\Sigma'$ samples a random walk from $x$ to $y$ and
accepts a proof $\psi'$ if
\begin{enumerate}
    \item
        $x$ and $y$ have \colorbox{yellow!75!white}{consistent opinions} about
        the value of each vertex appearing in the walk, and
    \item
        \colorbox{yellow!75!white}{any possible world derived from such opinions}
        satisfies all edges along the walk.
\end{enumerate}
(The difference from the original verifier is \colorbox{yellow!75!white}{highlighted}.)
Though the modified verifier preserves the perfect completeness,
it becomes ``much weaker'' than before, which makes it nontrivial to adapt
the soundness proof due to \cite{radhakrishnan2006gap,radhakrishnan2007dinurs}.
Consider, for example, the following question,
which is a simplified form of \cref{clm:amp:pow:consistent} used in the proof of \cref{lem:amp:pow}:
Given a pair of distributions $\vec{\nu}$ and $\vec{\mu}$ over $\Sigma'$ and
the \emph{most popular value} $\alpha_{\PLR}$ in \emph{non-squared} $\Sigma$ over $\vec{\nu}$; namely,
\begin{align}
    \alpha_{\PLR}
        \coloneq \argmax_{\alpha \in \Sigma}
        \Pr_{\alpha' \sim \vec{\nu}} \Bigl[\{\alpha\} \text{ and } \alpha' \text{ are consistent}\Bigr],
\end{align}
we would like to bound the following quantity $q$:\footnote{
This event represents one of the verifier's acceptance conditions.}
\begin{align}
    q \coloneq \Pr_{\alpha' \sim \vec{\nu}, \beta' \sim \vec{\mu}}
        \Bigl[\alpha' \text{ and } \beta' \text{ are consistent}\Bigr].
\end{align}
If the alphabet squaring trick were not applied (i.e., $\Sigma' = \Sigma$), consistency merely means equality, implying that
$q \leq \Pr_{\alpha' \sim \vec{\nu}}[\alpha' \text{ and } \alpha_{\PLR} \text{ are consistent}]$.
This is not correct for the squared alphabet case;\footnote{
For instance, let
$(\nu(\aaa), \nu(\bbb), \nu(\aaa\bbb)) = (0.51,0.49,0)$ and
$(\mu(\aaa), \mu(\bbb), \mu(\aaa\bbb)) = (0,0,1)$.
Then, $\alpha' \sim \vec{\nu}$ and $\beta' \sim \vec{\mu}$  are always consistent, though
$\alpha' \sim \vec{\nu}$ is consistent with $\alpha_{\PLR} = \aaa$ with probability $0.51$.
}
the actual inequality is
\begin{align}
    q \leq 2 \cdot \Pr_{\alpha' \sim \vec{\nu}}
        \Bigl[\alpha' \text{ and } \alpha_{\PLR} \text{ are consistent}\Bigr],
\end{align}
which stems from the fact that
\begin{align*}
    \forall \beta' \in \Sigma', \; \exists \beta_1, \beta_2 \in \Sigma, \text{ s.t. }
    \alpha' \text{ is consistent with } \beta'
    \implies
    \alpha' \text{ is consistent with } \beta_1 \text{ or } \beta_2.
\end{align*}

Other technical differences from
\citet{dinur2007pcp,radhakrishnan2006gap,radhakrishnan2007dinurs}
are summarized below.
\begin{itemize}
\item
    The popularity vote breaks reconfigurability:
    Similar to \cite{dinur2007pcp,radhakrishnan2006gap,radhakrishnan2007dinurs},
    given a reconfiguration sequence of proofs for $G'$,
    we apply the popularity vote to each proof in it
    and construct a sequence of assignments for $G$.
    However, the resulting sequence is \emph{invalid} in a sense that
    two neighboring assignments may differ in two or more vertices.
    In \cref{clm:amp:interpolate}, we \emph{interpolate} between every pair of neighboring assignments.
    
\item
    We give an explicit proof that if $G$ is expander, then so is the powered constraint graph $G'$ (\cref{lem:amp:pow:powered}).
    Though not required for the proof of the PCP theorem,
    this property is critical for our gap-preserving reduction to \SetCoverReconf.
    
\item
    The soundness error is improved from $\approx 0.9995$ \cite{radhakrishnan2007dinurs}
    to $0.9942$,
    resulting from the refinement of \cref{clm:amp:FE,lem:amp:pow:N1,lem:amp:pow:N2}.
    Such improvements make sense only for reconfiguration problems since
    parallel repetition of reconfiguration problems may not work as in the \NP regime
    \cite{ohsaka2023approximate}.
\end{itemize}

\subsection{Overview of Gap-preserving Reduction from \BCSPReconf to \SetCoverReconf}
The proof of \cref{thm:sc} constructs a gap-preserving reduction from
\MaxminBCSPReconf to \MinmaxSetCoverReconf.
Briefly recall
{Lund--Yannakakis}' reduction from \prb{Label Cover} to \prb{Set Cover} \cite{lund1994hardness}:
Given a constraint graph $G$ on vertex set $V$ and alphabet $\Sigma$,
we create a set system $(\calU, \calF)$, where
$\calF = \{S_{v,\alpha}\}_{v \in V, \alpha \in \Sigma}$ is a family of subsets of a universe $\calU$
indexed by vertex $v \in V$ and value $\alpha \in \Sigma$.
If a subfamily of $\calF$ covers $\calU$,
it must include $S_{v,\alpha}$ and $S_{w,\beta}$ for each edge $(v,w)$ of $G$
such that $(\alpha,\beta)$ satisfies $(v,w)$.
It turns out that given a cover $\calC$ for $(\calU,\calF)$,
we can generate a random assignment that satisfies
$\Omega\left(\frac{|V|^2}{|\calC|^2}\right)$-fraction
of edges in expectation.
Such \emph{randomness} of assignments is undesirable
for reducing between reconfiguration problems
since any two neighboring assignments in the reconfiguration sequence must differ in at most one vertex.
Our insight in resolving this difficulty is that
if a small cover $\calC$ of size $(2-\epsilon)|V|$ exists,
the value of $\approx \epsilon$-fraction vertices can be \emph{deterministically} chosen from $\calC$
so that all edges between them are satisfied.
Using the expander mixing lemma \cite{alon1988explicit},
we eventually bound the fraction of such edges from below by $\approx \epsilon^2$.

\subsection{Note on Alphabet Reduction of \BCSPReconf}
\label{subsec:intro:ABCreduct}
Given degree reduction \cite{ohsaka2023gap} and gap amplification (this paper)
for \MaxminBCSPReconf,
it is natural to ask the possibility of its \textbf{alphabet reduction} --- the last step of Dinur's proof \cite{dinur2007pcp} of the PCP theorem.
One barrier to achieving alphabet reduction for \MaxminBCSPReconf is its \textbf{robustization}.
In robustization of \prb{\BCSP} due to \cite{dinur2007pcp},
we replace each constraint $\pi_e$ by a Boolean circuit $C_e$
that accepts $\Enc(\alpha) \circ \Enc(\beta)$ such that $(\alpha, \beta)$ satisfies $\pi_e$,
where $\Enc$ is an error-correcting code.
The soundness case ensures that for many edges $e$,
the restricted assignment is $\Omega(1)$-far from any satisfying assignment to $C_e$,
so that assignment testers (of constant size) \cite{dinur2007pcp} can be applied.
However, applying such replacement to \MaxminBCSPReconf,
we need to reconfigure between a pair of codewords of $\Enc$,
passing through strings $\Omega(1)$-far from any codeword of $\Enc$.
This is problematic for the perfect completeness.
Currently, it is unclear whether this issue can be resolved, e.g., by the alphabet squaring trick.

\subsection{Additional Related Work}
Other optimization versions of reconfiguration problems include the following:
\prb{Subset Sum Reconfiguration} admits a PTAS~\cite{ito2014approximability},
\prb{Submodular Reconfiguration} admits a constant-factor approximation~\cite{ohsaka2022reconfiguration},
\prb{Clique Reconfiguration} and \prb{SAT Reconfiguration} are \NP-hard to approximate~\cite{ito2011complexity}.
Our results are fundamentally different from \cite{ito2011complexity} in that
we do not rely on known \NP-hardness results or the parallel repetition theorem.

We note that approximability of reconfiguration problems frequently refers to
that of \emph{the shortest reconfiguration sequence}
\cite{kaminski2011shortest,yamanaka2015swapping,miltzow2016approximation,bonnet2018complexity,bousquet2019shortest,bonamy2020shortest,bousquet2020approximating,ito2022shortest},
and there is a different type of optimization variant
called \emph{incremental optimization under the reconfiguration framework} \cite{ito2022incremental,blanche2020decremental,yanagisawa2021decremental},
both of which seem orthogonal to the present study.
Other gap amplification results are known for
\prb{Small-Set Expansion}~\cite{raghavendra2014gap} and
\prb{Parameterized Steiner Orientation}~\cite{wlodarczyk2020parameterized}.
Unlike them, gap amplification for reconfiguration problems
needs to sustain reconfigurability, which entails nontrivial modifications to
\cite{dinur2007pcp,radhakrishnan2006gap,radhakrishnan2007dinurs}.

For random instances,
the separation phenomena of the overlaps between near-optimal solutions 
has been investigated, called the \emph{overlap gap property} (OGP)
\cite{gamarnik2021overlap,achlioptas2011solution,mezard2005clustering,gamarnik2017limits,wein2021optimal}.
Consider an {Erd\H{o}s}--R\'{e}nyi graph over $n$ vertices,
in which each of the ${n \choose 2}$ potential edges is realized with probability $\frac{1}{2}$.
Then, for some $\nu_1, \nu_2 \in (0,1)$ with $\nu_1 < \frac{1}{2}\nu_2$,
for any independent sets $I$ and $J$ of size greater than $\alpha \cdot 2 \log_2 n$,
where $\alpha = \frac{1}{2} + \frac{1}{2\sqrt{2}}$,
$\frac{|I \triangle J|}{2\log_2 n}$ is in either $[0, \nu_1]$ or $[\nu_2, 1]$ (with high probability).
Though its motivation comes from the algorithmic hardness of random instances,
OGP implies approximate reconfigurability:
either
a trivial reconfiguration from $I$ to $J$ achieves $\frac{\alpha-\nu_1}{\alpha}$-factor approximation, or
we can declare that any reconfiguration induces an independent set of size at most $\alpha \cdot 2 \log_2 n$.
It is unclear if OGP could be used to derive
inapproximability results of reconfiguration problems in the worse-case sense.

\subsection{Follow-up Work}
\label{subsec:intro:follow-up}
After a preliminary version of this paper was presented at
the 35th Annual ACM-SIAM Symposium on Discrete Algorithms (SODA 2024) \cite{ohsaka2024gap},
the following progress has been made.
\citet{karthik2023inapproximability,hirahara2024probabilistically}
independently announced the proof of RIH,
implying that the $\PSPACE$-hardness results
shown in this paper hold \emph{unconditionally}.
\citet{karthik2023inapproximability} further
showed that \MinmaxSetCoverReconf is \NP-hard to approximate within a $(2-\epsilon)$-factor, and
demonstrated matching lower and upper bounds for \MaxminBCSPReconf, i.e.,
\NP-hardness of $\left(\frac{1}{2}+\epsilon\right)$-factor approximation
(which improves \cref{thm:NP-BCSPR})
and a $\left(\frac{1}{2}-\epsilon\right)$-factor approximation algorithm,
where $\epsilon > 0$ is any small number.
Subsequently, \citet{hirahara2024optimal} proved that
\MinmaxSetCoverReconf is \PSPACE-hard to approximate within a factor of 
$2-\frac{1}{\operatorname{polyloglog} n}$,
where $n$ is the universe size,
improving upon 
\citet{karthik2023inapproximability} and the present study.
In particular, both \MinmaxSetCoverReconf and \MinmaxDominatingSetReconf
do not admit $1.999$-factor approximation algorithms unless \cP~$=$~\PSPACE.
\citet{ohsaka2024alphabet} gave a reconfiguration analogue of \textbf{alphabet reduction} \`{a}~la~\citet{dinur2007pcp},
which resolves the issue discussed in \cref{subsec:intro:ABCreduct}
by using the reconfigurability of Hadamard codes.

\section{Preliminaries}
\label{sec:pre}
\paragraph{Notations.}
For any integer $n \in \bbN$, let $ [n] \coloneq \{1, 2, \ldots, n\} $.
For a statement $P$, $\llbracket P \rrbracket$ is $1$ if $P$ is true and $0$ otherwise.
A \emph{sequence} $\scrS$ of a finite number of objects $S^{(1)}, \ldots, S^{(\TTT)}$
is denoted by $( S^{(1)}, \ldots, S^{(\TTT)} )$, and
we write $S^{(\ttt)} \in \scrS$ to indicate that $S^{(\ttt)}$ appears in $\scrS$.
For a graph $G=(V,E)$, let $V(G)$ and $E(G)$
denote the vertex set $V$ and edge set $E$ of $G$, respectively.
The edge set $E$ may include self-loops and parallel edges; in particular, it is a multiset.
The symbol $\uplus$ is used to denote the additive union of multisets.
For a vertex $v$ of $G$, let $d_G(v)$ denote the degree of $v$.
For a vertex set $S \subseteq V(G)$, we write $G[S]$ for the subgraph of $G$
induced by $S$.
For any two subsets $S$ and $T$ of $V(G)$,
we define
$e_G(S,T)$ as the number of edges between $S$ and $T$;\footnote{
Edges included in $S \cap T$ are counted twice.
}
namely,
\begin{align}
    e_G(S,T)
    \coloneq \left|\Bigl\{ (v,w) \in S \times T \Bigm| (v,w) \in E(G) \Bigr\}\right|.
\end{align}

\paragraph{Expander Graphs.}
In \cref{sec:amp,sec:sc}, expander graphs play a crucial role.
For a regular graph $G$ and its adjacency matrix $\mat{A}$,
the \emph{second largest eigenvalue in absolute value}, denoted $\lambda(G)$, is defined as
\begin{align}
    \lambda(G) \coloneq
    \max_{\vec{x} \neq \vec{0}, \vec{x} \perp \vec{1}}
    \frac{\|\vec{A}\vec{x}\|_2}{\|\vec{x}\|_2},
    \text{ where }
    \|\vec{x}\|_2 \coloneq \sqrt{\sum_{i}x_i^2}
    \text{ and }
    \vec{1} \text{ is the all-ones vector}.
\end{align}
For any integer $d \in \bbN$ and number $\lambda < d$, a \emph{$(d,\lambda)$-expander graph}
is a $d$-regular graph $G$ such that $\lambda(G) \leq \lambda$.
A $(d,\lambda)$-expander graph is called \emph{Ramanujan} if $\lambda \leq 2\sqrt{d-1}$.
We will use the following lemma later,
whose proof is included in \cref{app:proof} for the sake of completeness.
\begin{lemma}
\label{lem:lambda}
    For any regular graphs $G$ and $H$
    with adjacency matrices $\mat{A}_G$ and $\mat{A}_H$,
    let $G \uplus H$ and $GH$ denote the graph whose adjacency matrix
    is respectively described by $\mat{A}_G + \mat{A}_H$ and $\mat{A}_G\mat{A}_H$.
    Then,
    $\lambda(G \uplus H) \leq \lambda(G) + \lambda(H)$ and
    $\lambda(GH) \leq \lambda(G) \cdot \lambda(H)$.
\end{lemma}

\paragraph{Reconfiguration Problems on Constraint Systems.}
We introduce reconfiguration problems on constraint systems.
The notion of constraint graphs is first introduced.

\begin{definition}
    A \emph{$q$-ary constraint graph} is defined as a tuple $G=(V,E,\Sigma,\Pi)$ such that
    \begin{itemize}
    \item 
        $(V,E)$ is a $q$-uniform hypergraph called the \emph{underlying graph} of $G$,
    \item 
        $\Sigma$ is a finite set called the \emph{alphabet}, and
    \item 
        $\Pi = (\pi_e)_{e \in E}$ is a collection of $q$-ary \emph{constraints}, where
        each $\pi_e \subseteq \Sigma^e$ is a set of $q$-tuples of acceptable values that
        $q$ vertices of $e$ can take.
    \end{itemize}
    A binary constraint graph is simply referred to as a \emph{constraint graph}.
\end{definition}
An \emph{assignment} for a $q$-ary constraint graph $G$
is a function $\psi \colon V \to \Sigma$ that assigns a value of $\Sigma$ to each vertex of $V$.
We say that $\psi$ \emph{satisfies} a hyperedge $e = \{v_1, \ldots, v_q\} \in E$ (or a constraint $\pi_e$) if
$\psi(e) \coloneq (\psi(v_1), \ldots, \psi(v_q)) \in \pi_e$, and
$\psi$ \emph{satisfies} $G$ if it satisfies all hyperedges of $G$.
Moreover, we say that $G$ is \emph{satisfiable} if some assignment satisfies $G$.
For two satisfying assignments $\psi^\ini$ and $\psi^\tar$ for $G$,
a \emph{reconfiguration sequence from $\psi^\ini$ to $\psi^\tar$}
is any sequence $(\psi^{(1)}, \ldots, \psi^{(T)})$ of assignments such that
$\psi^{(1)} = \psi^\ini$,
$\psi^{(T)} = \psi^\tar$, and
any two neighboring assignments differ in at most one vertex.
In the \prb{$q$-CSP Reconfiguration} problem,
for a (satisfiable) $q$-ary constraint graph $G$ and its two satisfying assignments $\psi^\ini$ and $\psi^\tar$,
we are asked to decide if there is a reconfiguration sequence of satisfying assignments for $G$ from $\psi^\ini$ to $\psi^\tar$.
Hereafter,
the suffix ``$_W$'' designates the restricted case that the alphabet size $|\Sigma|$ is integer $W \in \bbN$.

Since we are concerned with approximate reconfigurability,
we formulate an optimization version of \prb{$q$-CSP Reconfiguration} \cite{ito2011complexity,ohsaka2023gap},
which allows going through \emph{non-satisfying assignments}.
For a $q$-ary constraint graph $G=(V,E,\Sigma,\Pi)$ and an assignment $\psi \colon V \to \Sigma$,
its \emph{value} is defined as the fraction of hyperedges of $G$ satisfied by $\psi$; namely,
\begin{align}
    \val_G(\psi) \coloneq \frac{1}{|E|} \cdot \left|\Bigl\{e \in E \Bigm| \psi \text{ satisfies } e \Bigr\}\right|.
\end{align}
For a reconfiguration sequence
$\sqpsi = ( \psi^{(1)}, \ldots, \psi^{(\TTT)} )$,
let $\val_G(\sqpsi)$ denote the \emph{minimum fraction} of satisfied hyperedges
over all $\psi^{(\ttt)}$'s in $\sqpsi$; namely,
\begin{align}
    \val_G(\sqpsi) \coloneq \min_{\psi^{(\ttt)} \in \sqpsi} \val_G(\psi^{(\ttt)}).
\end{align}
In \prb{Maxmin $q$-CSP Reconfiguration},
we wish to maximize $\val_G(\sqpsi)$ subject to $\sqpsi = ( \psi^\ini, \ldots, \psi^\tar )$.
For two assignments $\psi^\ini, \psi^\tar \colon V \to \Sigma$ for $G$,
let $\val_G(\psi^\ini \reco \psi^\tar)$ denote the maximum value of $\val_G(\sqpsi)$
over all possible reconfiguration sequences $\sqpsi$ from $\psi^\ini$ to $\psi^\tar$; namely,
\begin{align}
    \val_G(\psi^\ini \reco \psi^\tar)
    \coloneq \max_{\sqpsi = ( \psi^\ini, \ldots, \psi^\tar )} \val_G(\sqpsi).
\end{align}
The ``gap version'' of \prb{Maxmin $q$-CSP Reconfiguration} is defined as follows.

\begin{problem}
\label{prb:gap-CSP}
    For every numbers $0 \leq s \leq c \leq 1$ and integer $q \in \bbN$,
    \prb{Gap$_{c,s}$ $q$-CSP Reconfiguration} requests to determine
    for a $q$-ary constraint graph $G$ and its two assignments $\psi^\ini$ and $\psi^\tar$,
    whether
    $\val_G(\psi^\ini \reco \psi^\tar) \geq c$ (the input is a \YES instance) or
    $\val_G(\psi^\ini \reco \psi^\tar) < s$ (the input is a \NO instance).
    Here, $c$ and $s$ are respectively called \emph{completeness} and \emph{soundness}.
\end{problem}\noindent
The case of $s = c = 1$ particularly reduces to \prb{$q$-CSP Reconfiguration}.
We can safely assume that $\psi^\ini$ and $\psi^\tar$ satisfy $G$ whenever $c=1$.
The \emph{Reconfiguration Inapproximability Hypothesis} (RIH)~\cite{ohsaka2023gap}
postulates that
\prb{Gap$_{1, 1-\epsilon}$ $q$-CSP$_W$ Reconfiguration} is \PSPACE-hard
for some $\epsilon \in (0,1)$ and $q, W \in \bbN$,
which has been recently proven by \citet{hirahara2024probabilistically,karthik2023inapproximability}.

\section{Gap Amplification of \MaxminBCSPReconf}
\label{sec:amp}

We prove the main result of this paper; that is,
\MaxminBCSPReconf over expander graphs is
\PSPACE-hard to approximate within a factor of $0.9942$ under RIH.

\begin{theorem}\label{thm:amp}
    Under RIH, for every number $\rho \in (0,1)$,
    there exist universal constants $W, d \in \bbN$ such that
    \prb{Gap$_{1,0.9942}$ \BCSPReconf[$_W$]}
    is \PSPACE-hard even if
    the underlying graph is restricted to $(d,\lambda)$-expander with
    $\frac{\lambda}{d} \leq \rho$.
\end{theorem}\noindent
The proof of \cref{thm:amp} is based on an adaptation of the first two steps of
Dinur's proof \cite{dinur2007pcp} of the PCP theorem.
In \cref{subsec:amp:expand,subsec:amp:pow},
we adapt the \textbf{expanderization} and \textbf{powering} steps
of \cite{dinur2007pcp}
to \MaxminBCSPReconf
(\cref{lem:amp:expand,lem:amp:pow}), respectively.
\cref{thm:amp} can be almost directly derived from
\cref{lem:amp:expand,lem:amp:pow} plus the \textbf{degree reduction} step of \cite{ohsaka2023gap},
whose details involving parameter tuning are described in \cref{subsec:amp:together}.

\subsection{Expanderization}
\label{subsec:amp:expand}
Here, we superimpose an expander graph on
a small-degree constraint graph of \MaxminBCSPReconf so that
the resulting constraint graph is expander.
The proof of the following lemma is reminiscent of that of \cite[Lemma 4.4]{dinur2007pcp}.

\begin{lemma}
\label{lem:amp:expand}
    For every
    number $\epsilon \in (0,1)$,
    positive integers $W, \Delta \in \bbN$, and
    positive even integer $d_0 \geq 3$,
    there exists a gap-preserving reduction from
    \prb{Gap$_{1,1-\epsilon}$ \BCSPReconf[$_W$]}
    whose underlying graph is $\Delta$-regular
    to
    \prb{Gap$_{1,1-\frac{\Delta}{\Delta+d_0}\epsilon}$ \BCSPReconf[$_W$]}
    whose underlying graph is $(\Delta + d_0, \Delta + 2\sqrt{d_0})$-expander.
\end{lemma}
\begin{proof}
Let $(\psi^\ini,\psi^\tar;G)$ be an instance of
\MaxminBCSPReconf[$_W$], where
its constraint graph $G=(V,E,\Sigma, \Pi=(\pi_e)_{e \in E})$ is $\Delta$-regular.
Construct a new constraint graph $G'=(V,E',\Sigma,\Pi'=(\pi'_{e'})_{e' \in E'})$
on the same vertex set $V$ and alphabet $\Sigma$ as follows:
\begin{description}
\item[\textbf{Edge set:}]
    Using an explicit construction for near-Ramanujan graphs \cite{mohanty2021explicit,alon2021explicit},
    create a $(d_0,2\sqrt{d_0})$-expander graph $X$ on vertex set $V$ in polynomial time.\footnote{
        We can safely assume that
        $|V|$ is sufficiently large so that the construction of 
        \cite{mohanty2021explicit,alon2021explicit} can be applied.
    }
    Define $E' \coloneq E \uplus E(X)$.
\item[\textbf{Constraints:}]
    The constraint $\pi'_{e'}$ for each edge $e' \in E'$ is defined as follows:
    \begin{itemize}
    \item
        If $e' \in E$, define $\pi'_{e'} \coloneq \pi_{e'}$ (i.e., it is the same as before).
    \item
        Otherwise (i.e., $e' \in E(X)$),
        define $\pi'_{e'} \coloneq \Sigma^{e'}$ (i.e., it is always satisfied).
    \end{itemize}
\end{description}\noindent
Observe that $G'$ is $(\Delta+d_0,\Delta+2\sqrt{d_0})$-expander
because it follows from \cref{lem:lambda} that
\begin{align}
    \lambda(G') \leq \lambda(G) + \lambda(X) \leq \Delta+2\sqrt{d_0}.
\end{align}
Since edges of $X$ are always satisfied,
any assignment $\psi \colon V \to \Sigma$ satisfies
\begin{align}
    \val_{G'}(\psi) = \frac{\Delta}{\Delta+d_0}\val_G(\psi) + \frac{d_0}{\Delta+d_0},
\end{align}
implying the completeness and soundness, as desired.
\end{proof}

\subsection{Powering}
\label{subsec:amp:pow}
Given an expander constraint graph,
we reduce its soundness error by using
the redesigned powering step of \cite{dinur2007pcp,radhakrishnan2006gap,radhakrishnan2007dinurs}
for \MaxminBCSPReconf, as formally stated below.

\begin{lemma}[Gap amplification lemma]
\label{lem:amp:pow}
    For every
    number $\epsilon \in (0,1)$,
    positive integers $W,d \in \bbN$, and
    number $\lambda > 0$ such that $\lambda \leq \frac{d}{2}$,
    there exists a gap-preserving reduction from
    \prb{Gap$_{1,1-\epsilon}$ \BCSPReconf[$_W$]}
    whose underlying graph is $(d,\lambda)$-expander
    to
    \prb{Gap$_{1,1-\epsilon'}$ \BCSPReconf[$_{W'}$]}
    whose underlying graph is $(d',\lambda')$-expander, where
    \begin{align}
        \epsilon' = \frac{0.0294}{3 + 2\frac{d}{d-\lambda}},
    \end{align}
    $d'$ and $\lambda'$ depend only on the values of $\epsilon$, $d$, and $\lambda$ such that
    $\frac{\lambda'}{d'} \leq \frac{4^{100}}{100} \cdot \frac{\lambda}{d}$, and
    $W' = W^{d^{\bigO\left(\frac{1}{\epsilon}\right)}}$.
\end{lemma}\noindent
Note that $\epsilon' \approx 0.0058$ if $\frac{\lambda}{d}\approx 0$.
The remainder of this subsection is devoted to the proof of \cref{lem:amp:pow}.

\subsubsection{PCP View of the Powering Step}
We first describe the powering step of \MaxminBCSPReconf
from the viewpoint of probabilistically checkable proof (PCP) systems, followed by its analysis.
Let $(\psi^\ini,\psi^\tar; G)$ be an instance of
\prb{Gap$_{1,1-\epsilon}$ \BCSPReconf[$_W$]}, where
$G=(V,E,\Sigma,\Pi=(\pi_e)_{e \in E})$ is a constraint graph such that
$(V,E)$ is $(d,\lambda)$-expander such that $\lambda \leq \frac{d}{2}$,
$|\Sigma| = W$,
$\psi^\ini$ and $\psi^\tar$ satisfy $G$.
Note that $\val_G(\psi^\ini \reco \psi^\tar)$ is either
equal to $1$ or less than $1-\epsilon$.
We first apply the \emph{alphabet squaring trick} \cite{ohsaka2023gap} to $\Sigma$ and define
\begin{align}
    \Sigma'
        \coloneq
        \Bigl\{ \{\alpha\} \Bigm| \alpha \in \Sigma \Bigr\}
        \cup
        \Bigl\{ \{\alpha, \beta\} \Bigm| \alpha \neq \beta \in \Sigma \Bigr\}.
\end{align}
Note that $|\Sigma'| = \frac{W(W+1)}{2}$.
We say that $\alpha'$ and $\beta'$ in $\Sigma'$ are \emph{consistent} if
``$\alpha' \subseteq \beta'$ or $\alpha' \supseteq \beta'$'' and
are \emph{conflicting} otherwise.
Such a notion of consistency will affect the design of the verifier's test and the popularity vote.
Below, we shall define the proof format and its verifier.

\paragraph{Proof Format and Verifier's Test.}
Let $\rep \coloneq \left\lceil \frac{2}{\epsilon} \right\rceil \in \bbN$ be
a parameter specifying the (expected) length of random walks,\footnote{
This parameter corresponds to ``$t$'' in 
\cite{dinur2007pcp,radhakrishnan2006gap,radhakrishnan2007dinurs,guruswami2005course}.
}
whose value will be justified in the proof of \cref{lem:amp:pow}.
Define $\Rep \coloneq 100\rep$.
The entire \emph{proof} is specified by a function
$\psi' \colon V \to \Sigma'^{d^{\Rep + 1}}$.
For each vertex $x$,
$\psi'(x)$ claims values (in $\Sigma'$) of all vertices within a distance of $\Rep$.\footnote{
The number of such vertices is at most $1+d+d^2+\cdots+d^{\Rep} \leq d^{\Rep + 1}$.
}
We use $\psi'(x)[v]$ to denote a value of $\Sigma'$ claimed to be assigned to a vertex $v$,
which can be thought of as \emph{$x$'s opinion on the value of $v$}.
If $v$ is more than $\Rep$ distance away from $x$, then $\psi'(x)[v]$ is \emph{undefined}.

The \emph{verifier} probabilistically reads a small portion of a given proof $\psi' \colon V \to \Sigma'^{d^{\Rep+1}}$ to decide
whether to accept it.
For an (undirected) edge $e = (v,w)$ of $G$,
a \emph{step} from $v$ to $w$ is written as $\vv{e} = \vv{vw}$.
For a pair of vertices $x$ and $y$ and a step $\vv{e} = \vv{vw}$ of $G$,
$\psi'(x)$ and $\psi'(y)$ \emph{pass the test at step $\vv{e}$} if
one of $\psi'(x)[v]$, $\psi'(x)[w]$, $\psi'(y)[v]$, or $\psi'(y)[w]$ is undefined, or
all of the following conditions hold:
\begin{enumerate}[label=(C\arabic*)]
\item
    $\psi'(x)[v]$ and $\psi'(y)[v]$ are consistent;
    \label{enm:amp:pow:C1}
\item
    $\psi'(x)[w]$ and $\psi'(y)[w]$ are consistent;
    \label{enm:amp:pow:C2}
\item
    $(\psi'(x)[v] \cup \psi'(y)[v]) \times (\psi'(x)[w] \cup \psi'(y)[w]) \subseteq \pi_e$.
    In other words, for all $\alpha \in \psi'(x)[v] \cup \psi'(y)[v]$ and $\beta \in \psi'(x)[w] \cup \psi'(y)[w]$, $(\alpha, \beta)$ satisfies $e$.
    \label{enm:amp:pow:C3}
\end{enumerate}
Otherwise, we say that $\psi'(x)$ and $\psi'(y)$ \emph{fail the test at $\vv{e}$}.
For the verifier to select $x$ and $y$, we
introduce the following two types of random walk \cite{guruswami2005course}.

\begin{definition}[\cite{guruswami2005course}]
    For a graph $G=(V,E)$,
    an \emph{after-stopping random walk} (ASRW) starting from a vertex $s \in V$ is
    obtained by the following procedure:
    \begin{itembox}[l]{\textbf{After-stopping random walk.}}
    \begin{algorithmic}[1]
        \State set $v_0 \coloneq s$.
        \For{$k = 1 \textbf{ to } \infty$}
            \State sample a random neighbor $v_{k}$ of $v_{k-1}$ and proceed along step $\vv{e_k} \coloneq \vv{v_{k-1}v_{k}}$.
            \State \textbf{break} with probability $\frac{1}{\rep}$.
        \EndFor
        \State \textbf{return} walk $\vec{W} = ( \vv{e_1}, \vv{e_2}, \ldots, \vv{e_k} )$.
    \end{algorithmic}
    \end{itembox}
    A \emph{before-stopping random walk} (BSRW) starting from vertex $s$ is
    obtained by the modified ASRW so that the order of line~3 and line~4 is reversed.
\end{definition}
Note that ASRWs have length at least $1$ while BSRWs can have length $0$.
The verifier makes a decision according to the following randomized algorithm:

\begin{itembox}[l]{\textbf{Verifier's test.}}
\begin{algorithmic}[1]
    \State choose a starting vertex $x$ from $V$ uniformly at random.
    \State sample an ASRW starting from $x$,
        denoted $\vec{W} = ( \vv{e_1}, \ldots, \vv{e_k} )$,
        which will be referred to as the \emph{verifier's walk};
        let $y$ be the ending terminal of $\vec{W}$.
    \For{\textbf{each} step $\vv{e_i}$ of $\vec{W}$}
        \If{$\psi'(x)$ and $\psi'(y)$ fail the test at $\vv{e_i}$}
            \State \textbf{declare} \textsf{reject}.
        \EndIf
    \EndFor
    \State \textbf{declare} \textsf{accept}.
\end{algorithmic}
\end{itembox}
\cref{fig:verifier} illustrates running examples of the verifier
when the original constraint graph $G=(V,E,\Sigma,\Pi)$
represents a \prb{3-Coloring} instance;\footnote{
We remark that a reconfiguration analogue of \prb{3-Coloring}
admits a polynomial-time algorithm \cite{cereceda2011finding}.
} i.e.,
$\Sigma \coloneq \{\RRR, \GGG, \BBB\}$ and 
$\pi_e \coloneq \Sigma^2 \setminus \{(\RRR,\RRR), (\GGG,\GGG), (\BBB,\BBB)\}$
for all $e \in E$.
Construct two proofs
$\psi'^\ini, \psi'^\tar \colon V \to \Sigma'^{d^{\Rep + 1}}$
from $\psi^\ini$ and $\psi^\tar$
such that
$\psi'^\ini(x)[v] \coloneq \{\psi^\ini(v)\}$ and
$\psi'^\tar(x)[v] \coloneq \{\psi^\tar(v)\}$
for all $v$ and $x$ within distance $\Rep$.
Observe that the verifier accepts both $\psi'^\ini$ and $\psi'^\tar$ with probability $1$,
finishing the description of the powering step (from a PCP point of view).

\begin{figure}[t]
\centering
\subfloat[
    Suppose $\psi'(x){[v]}=\{\RRR\}$ and $\psi'(y){[v]}=\{\BBB\}$.
    Since $\{\RRR\}$ and $\{\BBB\}$ are conflicting,
    $\psi'(x)$ violates \ref{enm:amp:pow:C1}.
]{\resizebox{.48\columnwidth}{!}{\begin{tikzpicture}
    \input{temp-verifier}
    \node[node, fill=\whitegreen] at (w) {$w$};
    \node[node] at (y) {$y$};
        
    \node[above=20mm of $(x)!.5!(v)$, centered, font=\LARGE]{$\psi'(x)[v]=\{\RRR\}$};
    \node[above=20mm of $(v)!.5!(y)$, centered, font=\LARGE]{$\psi'(y)[v]=\{\BBB\}$};
    \node[below=20mm of $(x)!.5!(w)$, centered, font=\LARGE]{$\psi'(x)[w]=\{\GGG\}$};
    \node[below=20mm of $(w)!.5!(y)$, centered, font=\LARGE]{$\psi'(y)[w]=\{\GGG\}$};
    
    \draw[edge, out=+75, in=+120, looseness=.6, draw=\blackred] (x) to (v);
    \draw[edge, out=-75, in=-120, looseness=.4, draw=\blackgreen] (x) to (w);
    \draw[edge, out=+105, in=+60, looseness=.4, draw=\blackblue] (y) to (v);
    \draw[edge, out=-105, in=-60, looseness=.6, draw=\blackgreen] (y) to (w);
    
    \node[centered, color=red, text opacity=.7] at (v) {{\fontsize{72pt}{72pt}\selectfont \ding{55}}};
\end{tikzpicture}}}
\hfill
\subfloat[
    Suppose $\psi'(x){[v]}=\{\RRR\}$, $\psi'(y){[v]}=\{\RRR,\BBB\}$, $\psi'(x){[w]}=\{\RRR\}$, and $\psi'(y){[w]}=\{\RRR\}$.
    Since $(\psi'(x){[v]} \cup \psi'(y){[v]}) \times (\psi'(x){[w]} \cup \psi'(y){[w]}) = \{\RRR,\BBB\} \times \{\RRR\}$,
    which contains $(\RRR,\RRR) \notin \pi_{e_i}$,
    $\psi'(x)$ and $\psi'(y)$ violate \ref{enm:amp:pow:C3}.
]{
\resizebox{.48\columnwidth}{!}{\begin{tikzpicture}
    \input{temp-verifier}
    \node[node, shading=axis, top color=\whitered, bottom color=\whiteblue] at (v) {$v$};
    \node[node, fill=\whitered] at (w) {$w$};
    \node[node] at (y) {$y$};
    
    \node[above=20mm of $(x)!.5!(v)$, centered, font=\LARGE]{$\psi'(x)[v]=\{\RRR\}$};
    \node[above=20mm of $(v)!.5!(y)$, centered, font=\LARGE]{$\psi'(y)[v]=\{\RRR,\BBB\}$};
    \node[below=20mm of $(x)!.5!(w)$, centered, font=\LARGE]{$\psi'(x)[w]=\{\RRR\}$};
    \node[below=20mm of $(w)!.5!(y)$, centered, font=\LARGE]{$\psi'(y)[w]=\{\RRR\}$};
    
    \draw[edge, out=+75, in=+120, looseness=.6, draw=\blackred] (x) to (v);
    \draw[edge, out=-75, in=-120, looseness=.4, draw=\blackred] (x) to (w);
    \draw[edge, out=+105, in=+60, looseness=.4, draw=\blackred] (y) to (v);
    \draw[edge, out=+115, in=+50, looseness=.4, draw=\blackblue] (y) to (v);
    \draw[edge, out=-105, in=-60, looseness=.6, draw=\blackred] (y) to (w);
    
    \node[below=2mm of $(v)!.5!(w)$, centered, color=red, text opacity=.7]{{\fontsize{72pt}{72pt}\selectfont \ding{55}}};
\end{tikzpicture}}}
\\
\subfloat[
    Suppose $\psi'(x){[v]}=\{\RRR\}$, $\psi'(y){[v]}=\{\RRR,\BBB\}$, $\psi'(x){[w]}=\{\GGG\}$, and $\psi'(y){[w]}=\{\GGG\}$.
    Observe that $\psi'(x)$ and $\psi'(y)$ satisfy \ref{enm:amp:pow:C1}--\ref{enm:amp:pow:C3};
    i.e., they pass the test at $e_i=(v,w)$.
]{
\resizebox{.48\columnwidth}{!}{\begin{tikzpicture}
    \input{temp-verifier}
    \node[node, shading=axis, top color=\whitered, bottom color=\whiteblue] at (v) {$v$};
    \node[node, fill=\whitegreen] at (w) {$w$};
    \node[node] at (y) {$y$};
    
    \node[above=20mm of $(x)!.5!(v)$, centered, font=\LARGE]{$\psi'(x)[v]=\{\RRR\}$};
    \node[above=20mm of $(v)!.5!(y)$, centered, font=\LARGE]{$\psi'(y)[v]=\{\RRR,\BBB\}$};
    \node[below=20mm of $(x)!.5!(w)$, centered, font=\LARGE]{$\psi'(x)[w]=\{\GGG\}$};
    \node[below=20mm of $(w)!.5!(y)$, centered, font=\LARGE]{$\psi'(y)[w]=\{\GGG\}$};
    
    \draw[edge, out=+75, in=+120, looseness=.6, draw=\blackred] (x) to (v);
    \draw[edge, out=-75, in=-120, looseness=.4, draw=\blackgreen] (x) to (w);
    \draw[edge, out=+105, in=+60, looseness=.4, draw=\blackred] (y) to (v);
    \draw[edge, out=+115, in=+50, looseness=.4, draw=\blackblue] (y) to (v);
    \draw[edge, out=-105, in=-60, looseness=.6, draw=\blackgreen] (y) to (w);
\end{tikzpicture}}}
\caption{
    Running examples of the verifier when 
    the original constraint graph $G=(V,E,\Sigma,\Pi)$ represents
    a \prb{$3$-Coloring} instance; i.e.,
    $\Sigma \coloneq \{\RRR, \GGG, \BBB\}$ and 
    $\pi_e \coloneq \Sigma^2 \setminus \{(\RRR,\RRR), (\GGG,\GGG), (\BBB,\BBB)\}$
    for all $e \in E$.
    Here, for a proof $\psi' \colon V \to \Sigma'^{d^{\Rep+1}}$,
    the verifier selects a random walk from $x$ to $y$, and
    is about to perform the test of $\psi'(x)$ and $\psi'(y)$ at $e_i \coloneq (v,w)$.
}
\label{fig:verifier}
\end{figure}

A \emph{reconfiguration sequence from $\psi'^\ini$ to $\psi'^\tar$}
is defined similarly to that for assignments (see \cref{sec:pre}).
We first show the completeness,
which is immediate from the alphabet squaring trick and the verifier's test.
\begin{lemma}
\label{lem:amp:pow:complete}
    If $\val_G(\psi^\ini \reco \psi^\tar) = 1$,
    there exists a reconfiguration sequence $\sqpsi'$ from $\psi'^\ini$ to $\psi'^\tar$ such that
    the verifier accepts every proof in $\sqpsi'$ with probability $1$.
\end{lemma}
\begin{proof}
It suffices to consider the case that $\psi^\ini$ and $\psi^\tar$ differ in exactly one vertex, say, $v \in V$.
Let $\alpha \coloneq \psi^\ini(v)$ and $\beta \coloneq \psi^\tar(v)$.
Observe that $\psi'^\ini(x) \neq \psi'^\tar(x)$ if and only if
$x$'s opinion on the value of $v$ is defined.
Define a new proof $\hat{\psi}' \colon V \to \Sigma'^{d^{\Rep+1}}$ as follows:
\begin{align}
    \hat{\psi}'(x)[w] \coloneq
    \begin{cases}
        \psi'^\ini(x)[w] = \psi'^\tar(x)[w] & \text{if } w \neq v,  \\
        \{\alpha, \beta\} & \text{if } w = v.
    \end{cases}
\end{align}
Consider a reconfiguration sequence $\sqpsi'$ from $\psi'^\ini$ to $\psi'^\tar$ \emph{via} $\hat{\psi}'$
obtained by the following procedure:

\begin{itembox}[l]{\textbf{Reconfiguration sequence $\sqpsi'$ from $\psi'^\ini$ to $\psi'^\tar$.}}
\begin{algorithmic}[1]
    \For{\textbf{each} $x$ within distance $\Rep$ from $v$}
        \State change $x$'s current value from $\psi'^\ini(x)$ to $\hat{\psi}'(x)$.
    \EndFor
    \For{\textbf{each} $x$ within distance $\Rep$ from $v$}
        \State change $x$'s current value from $\hat{\psi}'(x)$ to $\psi'^\tar(x)$.
    \EndFor
\end{algorithmic}
\end{itembox}
Let $\psi'^{\circ}$ be any intermediate proof of $\sqpsi'$.
Then,
\begin{align}
    \psi'^{\circ}(x)[w] \text{ is }
    \begin{cases}
        \text{equal to } \psi'^\ini(x)[w] = \psi'^\tar(x)[w] & \text{if } w \neq v,  \\
        \text{either of } \{\alpha\}, \{\beta\}, \text{ or } \{\alpha, \beta\} & \text{if } w = v.
    \end{cases}
\end{align}
Since $\psi'^\ini$ and $\psi'^\tar$ satisfy \ref{enm:amp:pow:C3}, so does $\psi'^\circ$.
Define $K$ as the set of values that $\psi'^{\circ}(x)[v]$ has taken over all $x$; namely,
\begin{align}
    K \coloneq \Bigl\{
        \psi'^{\circ}(x)[v] \in \Sigma' \Bigm|
        x \text{ and } v \text{ are within distance } \Rep
    \Bigr\},
\end{align}
Observe that $K$ falls into one of the following cases:
\begin{itemize}
    \item $\{\{\alpha\}\}$ at the beginning of line~1,
    \item $\{\{\alpha\}, \{\alpha, \beta\}\}$ during lines~1--2,
    \item $\{\{\alpha, \beta\}\}$ between lines~2 and 3,
    \item $\{\{\alpha, \beta\}, \{\beta\}\}$ during lines~3--4, or
    \item $\{\{\beta\}\}$ at the ending of line~4.
\end{itemize}
In particular,
$\{\{\alpha\}, \{\beta\}\}$ does not appear, implying that
$\psi'^\circ(x)[v]$ and $\psi'^\circ(y)[v]$ are always consistent
for any $x$ and $y$.
Therefore, \ref{enm:amp:pow:C1} and \ref{enm:amp:pow:C2} must be satisfied by $\psi'^\circ$.
Consequently,
the verifier must accept $\psi'^{\circ}$ with probability $1$, as desired.
\end{proof}

The technical part of this subsection is the proof of the soundness, stated below.
\begin{lemma}
\label{lem:amp:pow:sound}
    If $\val_G(\psi^\ini \reco \psi^\tar) < 1-\epsilon$,
    every reconfiguration sequence from $\psi'^\ini$ to $\psi'^\tar$ includes a proof
    that is rejected by the verifier with probability at least
    $\Omega(\rep \epsilon)$.
\end{lemma}

\paragraph{Popularity Vote, Faulty Steps, and Soundness Proof.}
Hereafter, we assume that $\val_G(\psi^\ini \reco \psi^\tar) < 1-\epsilon$.
Given a proof $\psi' \colon V \to \Sigma'^{d^{\Rep + 1}}$,
we explain how to extract from it a new assignment $\PLR(\psi') \colon V \to \Sigma$ for $G$.
Our ``popularity vote'' is obtained by incorporating the consistency over $\Sigma'$ into that of \citet{dinur2007pcp}.
Consider a BSRW $\vec{W}$ starting from vertex $v \in V$ conditioned on its length being at most $\Rep-1$.
Denoting the ending terminal of $\vec{W}$ by $y$, we find $\psi'(y)[v]$ a distribution over $\Sigma'$.
Then, $\PLR(\psi')(v)$ is defined as
the \emph{most popular value} of $\Sigma$ over this distribution;
namely,\footnote{
Ties are broken according to any prefixed order of $\Sigma$.
}
\begin{align}
    \PLR(\psi')(v) \coloneq \argmax_{\alpha \in \Sigma}
        \Pr_{\text{BSRW } \vec{W} \text{ from } v}
        \Bigl[ \alpha \in \psi'(y)[v] \Bigm| \text{length of } \vec{W} \leq \Rep-1 \Bigr].
\end{align}
Suppose we are given a reconfiguration sequence from $\psi'^\ini$ to $\psi'^\tar$,
denoted $\sqpsi' = ( \psi'^{(1)}, \ldots, \psi'^{(\TTT)} )$,
such that the minimum acceptance probability is maximized.
Then, we construct a sequence
$\sqpsi = ( \psi^{(1)}, \ldots, \psi^{(\TTT)} )$
of assignments for $G$
such that
$\psi^{(\ttt)} \coloneq \PLR(\psi'^{(\ttt)})$ for all $\ttt \in [\TTT]$.
In particular,
$\psi^{(1)} = \PLR(\psi'^\ini) = \psi^\ini$
and 
$\psi^{(\TTT)} = \PLR(\psi'^\tar) = \psi^\tar$.
Unfortunately, $\sqpsi$ is \emph{not} a valid reconfiguration sequence 
because two neighboring assignments $\psi^{(\ttt)}$ and $\psi^{(\ttt+1)}$
may differ in two or more vertices.
Using the fact that they differ in a small number of vertices,
we find that some $\psi^{(\ttt)}$ must violate an $(\epsilon - o(1))$-fraction of the edges.

\begin{claim}
\label{clm:amp:interpolate}
    There exists $\psi^{(\ttt)}$ in $\sqpsi$ such that
    $\val_G(\psi^{(\ttt)}) < 1-\delta$, where
    $\delta \coloneq \epsilon - \frac{d^{\Rep + 2}}{|V|}$.
\end{claim}
\begin{proof}
Since $\psi'^{(\ttt)}$ and $\psi'^{(\ttt+1)}$ differ in (at most) one vertex, 
$\psi^{(\ttt)}$ and $\psi^{(\ttt+1)}$ differ in at most $d^{\Rep + 1}$ vertices.
Thus, a valid reconfiguration sequence $\sqpsi^{(\ttt)}$ from $\psi^{(\ttt)}$ to $\psi^{(\ttt+1)}$
can be obtained by merely changing the value of at most $d^{\Rep + 1}$ vertices
(in which $\psi^{(\ttt)}$ and $\psi^{(\ttt+1)}$ differ).
Let $\sqpsi$ be a sequence obtained by concatenating such $\sqpsi^{(\ttt)}$ for all $\ttt \in [\TTT-1]$.
Since $\sqpsi$ is a valid reconfiguration sequence from $\psi^\ini$ to $\psi^\tar$,
we have $\val_G(\sqpsi) < 1-\epsilon$;
in particular, there exists an assignment $\psi^\circ$ in $\sqpsi$
such that $\val_G(\psi^\circ) < 1-\epsilon$.

Suppose $\psi^\circ$ appears in $\sqpsi^{(\ttt)}$.
Note that if two assignments for $G$ differ in at most one vertex,
their values differ by at most $\frac{d}{|E|}$.
Since $\psi^\circ$ and $\psi^{(\ttt)}$ differ in at most $d^{\Rep + 1}$ vertices,
we derive
\begin{align}
\begin{aligned}
    & \val_G(\psi^\circ)
    \geq \val_G(\psi^{(\ttt)}) - \frac{d}{|E|} \cdot d^{\Rep + 1} \\
    \therefore\;\; & \val_G(\psi^{(\ttt)})
    \leq \val_G(\psi^\circ) + \frac{d^{\Rep + 2}}{|E|} 
    < 1-\epsilon + \frac{d^{\Rep + 2}}{|V|},
\end{aligned}
\end{align}
which completes the proof.
\end{proof}
Let $\psi^{(\ttt)}$ be any ``bad'' assignment in a sense that
$\val_G(\psi^{(\ttt)}) < 1-\delta$,
as promised by \cref{clm:amp:interpolate}.
We would like to show that the verifier rejects $\psi'^{(\ttt)}$ with a sufficiently large probability.
Hereafter, we denote
$\psi \coloneq \psi^{(\ttt)}$ and $\psi' \coloneq \psi'^{(\ttt)}$
for notational convenience,
and fix a set $F \subset E$ of edges violated by $\psi$ by the following claim:\footnote{
    This is a slight improvement over the argument of \cite{radhakrishnan2006gap,radhakrishnan2007dinurs}.
}

\begin{claim}
\label{clm:amp:FE}
    There exists a set $F \subset E$ of edges violated by $\psi$ such that
    \begin{align}
    \label{eq:amp:FE}
        \min\left\{ \delta, \frac{1}{\rep}-\frac{1}{|E|} \right\} \leq \frac{|F|}{|E|} \leq \frac{1}{\rep}.
    \end{align}
\end{claim}
\begin{proof}
Let $F$ be the set of all edges of $G$ violated by $\psi$;
it holds that $|F| \geq \lceil \delta |E| \rceil$.
We slightly modify $F$ as follows:
\begin{itemize}
\item
    If $\delta \leq \frac{1}{\rep} - \frac{1}{|E|}$,
    we throw away some edges of $F$ so that
    $|F| = \lceil \delta |E| \rceil$,
    implying that
    \begin{align}
        \delta |E|
        \leq |F|
        = \lceil \delta |E| \rceil
        \leq \left\lceil \frac{|E|}{\rep}-1 \right\rceil
        \leq \frac{|E|}{\rep}.
    \end{align}
\item
    Otherwise,
    we throw away some edges of $F$ so that
    $|F| = \left\lceil \left(\frac{1}{\rep}-\frac{1}{|E|} \right)|E| \right\rceil$,
    implying that
    \begin{align}
        \left(\frac{1}{\rep}-\frac{1}{|E|}\right)|E|
        \leq |F|
        = \left\lceil \left(\frac{1}{\rep}-\frac{1}{|E|}\right)|E| \right\rceil
        = \left\lceil \frac{|E|}{\rep} - 1 \right\rceil
        \leq \frac{|E|}{\rep}.
    \end{align}
\end{itemize}
Consequently, we must have
\begin{align}
    \min\left\{ \delta, \frac{1}{\rep}-\frac{1}{|E|} \right\} \leq \frac{|F|}{|E|} \leq \frac{1}{\rep},
\end{align}
as desired.
\end{proof}

We then introduce the notion of \emph{faulty steps} \cite{dinur2007pcp,radhakrishnan2006gap,radhakrishnan2007dinurs}.
\begin{definition}[\cite{dinur2007pcp,radhakrishnan2006gap,radhakrishnan2007dinurs}]
    A step $\vv{e}$ of $G$ appearing in the verifier's walk is said to be \emph{faulty} if
    $e \in F$ and $\psi'(x)$ and $\psi'(y)$ fail the test at $\vv{e}$, where
    $x$ and $y$ are the terminals of the verifier's walk.
    Let $N$ denote a random variable for the number of faulty steps in the verifier's walk given a proof $\psi'$.
\end{definition}
Since the verifier rejects $\psi'$ whenever it encounters a faulty step,
it suffices to show that event ``$N > 0$'' is likely to occur.
Thus, we shall use the second moment method.
We first estimate the first moment of $N$,
whose proof requires a refined analysis of
\cite{radhakrishnan2006gap} and \cite[Claim 5.10]{radhakrishnan2007dinurs}.

\begin{lemma}
\label{lem:amp:pow:N1}
    It holds that
    \begin{align}
        \E[N] \geq \rep \frac{|F|}{|E|} \cdot \approxone^2 \cdot (\sqrt{2}-1)^2,
    \end{align}
    where $\approxone \coloneq 1-\rme^{-100}$.
\end{lemma}\noindent
The proof of \cref{lem:amp:pow:N1} resorts to the following property of
the verifier's walk \cite{radhakrishnan2007dinurs,guruswami2005course}:

\begin{lemma}[Verifier's walk lemma \cite{guruswami2005course,radhakrishnan2007dinurs}]
\label{lem:amp:pow:verifier-walk}
    Let $G$ be a regular graph and $\vv{e}=\vv{vw}$ be a step of $G$.
    Consider an ASRW starting from a vertex chosen from $V(G)$ uniformly at random
    conditioned on the event that step $\vv{e}$ appears exactly $k$ times for some integer $k \in \bbN$.
    Then, the terminals of the walk, denoted $x$ and $y$, are pairwise independent.
    Moreover, the distribution of $x$ and $y$ is the same as that of
    the ending terminal of a BSRW starting respectively from $v$ and $w$.
\end{lemma}
At the heart of the proof of \cref{lem:amp:pow:N1} is \cref{clm:amp:pow:consistent},
which (for fixed $v$) gives a bound of
$\Pr_{x,y}[\psi'(x)[v] \allowbreak \text{ and } \allowbreak \psi'(y)[v] \allowbreak \text{ are consistent}]$
using
$\Pr_{x}[\psi'(x)[v] \allowbreak \text{ and } \allowbreak \PLR(\psi')(v) \allowbreak \text{ are consistent}]$.

\begin{proof}[Proof of \cref{lem:amp:pow:N1}]
Fix an edge $e = (v,w)$ of $F$ and a step $\vv{e} = \vv{vw}$.
Let $N_{\vv{e}}$ denote a random variable for
the number of faulty occurrences of $\vv{e}$ in the verifier's walk.
Here, we condition on the event that $\vv{e}$ appears in the verifier's walk exactly $k$ times for some $k \in \bbN$.
By \cref{lem:amp:pow:verifier-walk},
the terminals $x$ and $y$ of the verifier's walk can be considered
generated by a pair of independent BSRWs starting from $v$ and $w$, respectively.
Thus, $\psi'(x)$ and $\psi'(y)$ fail the test at $\vv{e}$
(i.e., $\vv{e}$ is faulty)
if both BSRWs terminate within $\Rep-1$ steps
(so that $\psi'(x)[v]$, $\psi'(x)[w]$, $\psi'(y)[v]$, and $\psi'(y)[w]$ are all defined), and
one of the following events holds:
\begin{enumerate}[label=(E\arabic*)]
    \item
        $\psi'(x)[v]$ and $\psi'(y)[v]$ are conflicting;
        \label{lem:amp:pow:N1:E1}
    \item
        $\psi'(x)[w]$ and $\psi'(y)[w]$ are conflicting;
        \label{lem:amp:pow:N1:E2}
    \item
        $(\psi'(x)[v] \cup \psi'(y)[v]) \times (\psi'(x)[w] \cup \psi'(y)[w]) \not\subseteq \pi_e$.
        \label{lem:amp:pow:N1:E3}
\end{enumerate}
Observe that
each of the BSRWs terminates within $\Rep-1$ steps
with probability at least $1 - \left(1-\frac{1}{\rep}\right)^{100 \rep} \geq 1-\rme^{-100} = \approxone$.

Let $p_v$ (resp.~$p_w$) be the probability that
a BSRW starting from $v$ (resp.~$w$) terminates within $\Rep-1$ steps \emph{and}
$\psi'(x)[v] \ni \psi(v)$ (resp.~$\psi'(y)[w] \ni \psi(w)$);
i.e., $\psi'(x)[v]$ (resp.~$\psi'(y)[w]$) is consistent with the popularity vote.
To bound the probability of \ref{lem:amp:pow:N1:E1} and \ref{lem:amp:pow:N1:E2}, we use \cref{clm:amp:pow:consistent},
which is proven afterwards.

\begin{claim}
\label{clm:amp:pow:consistent}
    Conditioned on the event that ``$\vv{e} \text{ appears } k \text{ times}$,''
    the following hold\textup{:}
    \begin{align}
        \label{eq:amp:consistent:v}
        \Pr_{x,y}\Bigl[\psi'(x)[v] \text{ and } \psi'(y)[v] \text{ are consistent} \Bigr] & \leq 2p_v, \\
        \label{eq:amp:consistent:w}
        \Pr_{x,y}\Bigl[\psi'(x)[w] \text{ and } \psi'(y)[w] \text{ are consistent} \Bigr] & \leq 2p_w.
    \end{align}
\end{claim}\noindent
By \cref{clm:amp:pow:consistent},
\ref{lem:amp:pow:N1:E1} and \ref{lem:amp:pow:N1:E2} occur with probability at least
$\approxone \cdot (\approxone - 2p_v)$ and $\approxone\cdot(\approxone - 2p_w)$, respectively.
Observe further that \ref{lem:amp:pow:N1:E3} occurs
if $\psi'(x)[v] \ni \psi(v)$ and $\psi'(y)[w] \ni \psi(w)$ because
$(\psi(v), \psi(w)) \notin \pi_e$,
which holds with probability $p_v p_w$.
Thus,
the test at $\vv{e}$ fails with probability at least
\begin{align}
\label{eq:amp:consistent:failure}
    \max\Bigl\{\approxone\cdot (\approxone - 2p_v), \approxone\cdot(\approxone - 2p_w), p_v p_w \Bigr\}
    \geq \approxone^2 \cdot (\sqrt{2}-1)^2,
\end{align}
where we have used the fact that the left-hand side achieves the minimum when
$\approxone\cdot(\approxone - 2p_v) = \approxone\cdot(\approxone - 2p_w) = p_v p_w$,
implying that
$p_v = p_w = \approxone\cdot (\sqrt{2}-1)$.
Since the $k$ occurrences of $\vv{e}$ would be all faulty whenever the test at $\vv{e}$ fails,
we obtain
\begin{align}
\begin{aligned}
    \E[N_{\vv{e}}]
    & = \sum_{k \geq 0} k \cdot \Pr\Bigl[N_{\vv{e}} = k\Bigr] \\
    & = \sum_{k \geq 0} k \cdot
        \Pr\Bigl[\vv{e} \text{ appears } k \text{ times}\Bigr] \cdot
        \Pr\Bigl[\text{test at } \vv{e} \text{ fails} \Bigm| \vv{e} \text{ appears } k \text{ times}\Bigr] \\
    & \underbrace{\geq}_{\text{\cref{eq:amp:consistent:failure}}}
        \left(\sum_{k \geq 0} k \cdot \Pr\Bigl[\vv{e} \text{ appears } k \text{ times}\Bigr] \right)
            \cdot \approxone^2 \cdot (\sqrt{2}-1)^2 \\
    & = \E\Bigl[\text{\# appearances of } \vv{e}\Bigr] \cdot \approxone^2 \cdot(\sqrt{2}-1)^2 \\
    & \geq \frac{\rep}{2|E|} \cdot \approxone^2 \cdot (\sqrt{2}-1)^2.
\end{aligned}
\end{align}
Taking the sum over all edges of $F$, we derive
\begin{align}
    \E[N]
    = \sum_{(v,w) \in F} \E[N_{\vv{vw}}] + \E[N_{\vv{wv}}]
    \geq \rep \frac{|F|}{|E|} \cdot \approxone^2 \cdot(\sqrt{2}-1)^2,
\end{align}
completing the proof.
\end{proof}

\begin{proof}[Proof of \cref{clm:amp:pow:consistent}]
We show \cref{eq:amp:consistent:v} first.
Since $\psi'(x)[v]$ and $\psi'(y)[v]$ are pairwise independent
because of the verifier's walk lemma (\cref{lem:amp:pow:verifier-walk}),
we have
\begin{align}
\label{eq:amp:consistent:expand}
\begin{aligned}
    & \Pr_{x,y}\Bigl[ \psi'(x)[v] \text{ and } \psi'(y)[v] \text{ are consistent} \Bigr] \\
    & = \sum_{\alpha', \beta' \in \Sigma'}
    \Pr_{x}\Bigl[\psi'(x)[v] = \alpha'\Bigr] \cdot \Pr_{y}\Bigl[\psi'(y)[v] = \beta'\Bigr] \cdot
    \Bigl\llbracket \alpha' \text{ and } \beta' \text{ are consistent} \Bigr\rrbracket \\
    & = \sum_{\beta' \in \Sigma'} \Pr_{y}\Bigl[\psi'(y)[v] = \beta'\Bigr] \cdot
    \underbrace{
    \sum_{\alpha' \in \Sigma'} \Pr_{x}\Bigl[\psi'(x)[v] = \alpha'\Bigr] \cdot
    \Bigl\llbracket \alpha' \text{ and } \beta' \text{ are consistent} \Bigr\rrbracket
    }_{\clubsuit}.
\end{aligned}
\end{align}
Since \cref{eq:amp:consistent:expand} is a convex combination of $\clubsuit$'s,
letting $\vec{\mu} \in \Delta^{\Sigma'}$ be an arbitrary distribution over $\Sigma'$,
we can rewrite it as follows:
\begin{align}
\label{eq:amp:consistent:expand2}
\begin{aligned}
    \text{\cref{eq:amp:consistent:expand}}
    & \leq \max_{\vec{\mu} \in \Delta^{\Sigma'}}
    \sum_{\beta' \in \Sigma'} \mu(\beta') \cdot \sum_{\alpha' \in \Sigma'}
    \Pr_{x}\Bigl[ \psi'(x)[v] = \alpha' \Bigr] \cdot
    \Bigl\llbracket \alpha' \text{ and } \beta' \text{ are consistent} \Bigr\rrbracket \\
    & = \max_{\beta' \in \Sigma'} \sum_{\alpha' \in \Sigma'}
    \Pr_{x}\Bigl[\psi'(x)[v] = \alpha'\Bigr] \cdot
    \Bigl\llbracket \alpha' \subseteq \beta' \text{ or } \alpha' \supseteq \beta' \Bigr\rrbracket,
\end{aligned}
\end{align}
where the last equality holds because
the maximum value is attained when $\vec{\mu}$ is a one-point distribution.
We now bound \cref{eq:amp:consistent:expand2}
by case analysis on the size of $\beta'$:
\begin{description}
    \item[(Case 1)] 
        $|\beta'| = 1$: We have 
        \begin{align}
        \begin{aligned}
            \sum_{\alpha' \in \Sigma'} \Pr_{x}\Bigl[\psi'(x)[v] = \alpha'\Bigr] \cdot
            \Bigl\llbracket \alpha' \subseteq \beta' \text{ or } \alpha' \supseteq \beta' \Bigr\rrbracket
            & = \Pr_{x}\Bigl[\psi'(x)[v] \supseteq \beta'\Bigr] \\
            & \leq \Pr_{x}\Bigl[\psi'(x)[v] \ni \psi(v)\Bigr] = p_v,
        \end{aligned}
        \end{align}
        where the second last inequality is because
        $\psi(v)$ was determined by the popularity vote on $\psi'(x)[v]$.
    \item[(Case 2)]
        $|\beta'| = 2$: Letting $\beta' = \{\beta_1, \beta_2\}$ for some $\beta_1, \beta_2 \in \Sigma$,
        we obtain
        \begin{align}
        \begin{aligned}
            \sum_{\alpha' \in \Sigma'} \Pr_{x}\Bigl[\psi'(x)[v] = \alpha'\Bigr] \cdot
            \Bigl\llbracket \alpha' \subseteq \beta' \text{ or } \alpha' \supseteq \beta' \Bigr\rrbracket
            & = \Pr_{x}\Bigl[\psi'(x)[v] \subseteq \{\beta_1, \beta_2\}\Bigr] \\
            & \leq \Pr_{x}\Bigl[\psi'(x)[v] \ni \beta_1\Bigr] + \Pr_{x}\Bigl[\psi'(x)[v] \ni \beta_2\Bigr] \\
            & \leq \Pr_{x}\Bigl[\psi'(x)[v] \ni \psi(v)\Bigr] + \Pr_{x}\Bigl[\psi'(x)[v] \ni \psi(v)\Bigr] \\
            & \leq p_v + p_v = 2p_v.
        \end{aligned}
        \end{align}
\end{description}
Accordingly, we obtain
\begin{align}
    \Pr_{x,y}\Bigl[\psi'(x)[v] \text{ and } \psi'(y)[v] \text{ are consistent}\Bigr] \leq 2p_v.
\end{align}
We can show \cref{eq:amp:consistent:w} similarly, completing the proof.
\end{proof}

We then bound the second moment of $N$,
whose proof is almost the same as those of
\cite{radhakrishnan2006gap} and
\cite[Claim 5.10]{radhakrishnan2007dinurs}.

\begin{lemma}
\label{lem:amp:pow:N2}
    It holds that
    \begin{align}
        \E[N^2] \leq \rep \frac{|F|}{|E|} \cdot \left(3 + 2\frac{d}{d-\lambda}\right).
    \end{align}
\end{lemma}
Note that \cref{lem:amp:pow:N1,lem:amp:pow:N2} imply \cref{lem:amp:pow:sound}.
Before proving \cref{lem:amp:pow:N2},
we introduce the expander random walk lemma \cite{dinur2007pcp,guruswami2005course,radhakrishnan2007dinurs}.
\begin{lemma}[Expander random walk lemma \cite{dinur2007pcp,guruswami2005course,radhakrishnan2007dinurs}]
\label{lem:amp:pow:expander-walk}
    For a $(d,\lambda)$-expander graph $G = (V,E)$ and 
    a set $F \subseteq E$ of edges without self-loops,
    let $\vec{W} = ( \vv{e_1}, \vv{e_2}, \ldots )$ denote
    an ASRW starting from a vertex chosen from $V$ uniformly at random.
    Then, for any $j > i \geq 1$, it holds that
    \begin{align}
        \Pr \Bigl[e_j \in F \Bigm| e_i \in F \Bigr] \leq
        \left(1-\frac{1}{\rep}\right)^{j-i} \left(\frac{|F|}{|E|} + \left(\frac{\lambda}{d}\right)^{j-i-1}\right).
    \end{align}
\end{lemma}

\begin{proof}[Proof of \cref{lem:amp:pow:N2}]
Let $\chi_i$ denote a random variable $\llbracket e_i \in F \rrbracket$,
where $e_i$ is the \nth{$i$} edge of the verifier's walk.
Since the number of faulty steps in the verifier's walk is
bounded by the number of edges of $F$ appearing in the verifier's walk,
we have
\begin{align}
\begin{aligned}
    \E[N^2]
    & \leq \E\Bigl[(\text{\# edges of } F \text{ appearing in verifier's walk})^2\Bigr] \\
    & = \E\left[\left(\sum_{i \geq 1} \chi_i\right)^2 \right] \\
    & = \E\left[\sum_{i \geq 1} \chi_i \cdot \sum_{j \geq 1} \chi_j \right] \\
    & = \E\left[\sum_{i \geq 1} \chi_i^2\right]
    + 2 \cdot \E\left[\sum_{1 \leq i < j} \chi_i \chi_j\right].
\end{aligned}
\end{align}
Using that $\chi_i^2 = \chi_i$ and
$\E[\chi_i] = \Pr[\chi_i = 1] = \left(1-\frac{1}{\rep}\right)^{i-1} \frac{|F|}{|E|}$,
we obtain
\begin{align}
    \E\left[\sum_{i \geq 1} \chi_i^2\right]
    = \sum_{i \geq 1} \E[\chi_i]
    = \sum_{i \geq 1}\left(1-\frac{1}{\rep}\right)^{i-1} \frac{|F|}{|E|}
    = \rep \frac{|F|}{|E|}.
\end{align}
By \cref{lem:amp:pow:expander-walk} and 
$\E[\chi_i \chi_j] = \Pr[\chi_i = 1 \wedge \chi_j = 1]$, we derive
\begin{align}
\begin{aligned}
    \E\left[\sum_{1 \leq i < j} \chi_i \chi_j\right]
    & = \sum_{i \geq 1}\Pr\Bigl[\chi_i = 1\Bigr] \cdot
        \sum_{j > i}\Pr\Bigl[\chi_j = 1 \Bigm| \chi_i = 1\Bigr] \\
    & \leq \sum_{i \geq 1} \Pr\Bigl[\chi_i=1\Bigr] \cdot
        \sum_{j-i \geq 1} \left(1-\frac{1}{\rep}\right)^{j-i} \left[\frac{|F|}{|E|} + \left(\frac{\lambda}{d}\right)^{j-i-1}\right] \\
    & \leq \sum_{i \geq 1} \Pr\Bigl[\chi_i = 1\Bigr] \cdot
        \left(1-\frac{1}{\rep}\right) \cdot
        \sum_{k \geq 1}\left[ \frac{|F|}{|E|} \left(1-\frac{1}{\rep}\right)^{k-1} + \left(\frac{\lambda}{d}\right)^{k-1} \right] \\
    & = \rep \frac{|F|}{|E|} \cdot \left(1-\frac{1}{\rep}\right) \cdot \left[ \frac{|F|}{|E|} \rep + \frac{1}{1-\frac{\lambda}{d}} \right] \\
    & \underbrace{\leq}_{\text{\cref{clm:amp:FE}}}
        \rep\frac{|F|}{|E|} \cdot \left(1 + \frac{d}{d-\lambda}\right).
\end{aligned}
\end{align}
Consequently, we obtain 
\begin{align}
    \E[N^2]
    \leq \rep \frac{|F|}{|E|} + 2\cdot \rep \frac{|F|}{|E|} \cdot \left(1 + \frac{d}{d-\lambda}\right)
    = \rep \frac{|F|}{|E|} \cdot \left(3 + 2\frac{d}{d-\lambda}\right),
\end{align}
completing the proof.
\end{proof}

\subsubsection{Truncated Verifier's Walks and Powered Constraint Graphs}
\label{sec:amp:truncated}
Here, we approximately represent the verifier's walks by a constraint graph
to accomplish the proof of \cref{lem:amp:pow}.
Since there are infinitely many verifier's walks,
we first truncate them without significantly sacrificing the rejection probability.
Consider the \emph{truncated verifier's walk},
which always accepts the proof whenever
it follows more than $\Rep = 100\rep$ steps.
Let $N'$ denote a random variable for the number of faulty steps in the truncated verifier's walk $\vec{W}$; namely,
\begin{align}
    N' \coloneq N \cdot \Bigl\llbracket \text{length of } \vec{W} \leq \Rep \Bigr\rrbracket.
\end{align}
Since $N' \leq N$, it holds that
$\E[N'^2] \leq \E[N^2]$,
which can be bounded by \cref{lem:amp:pow:N2}.
Subsequently, we show that $\E[N']$ is not much smaller than $\E[N]$,
whose proof essentially follows that of
\cite[Lemma~2.3]{radhakrishnan2006gap} and \cite[Lemma~5.12]{radhakrishnan2007dinurs}.

\begin{lemma}
\label{lem:amp:pow:truncated-N1}
    It holds that
    \begin{align}
        \E[N'] \geq 0.1715 \cdot \rep \frac{|F|}{|E|}.
    \end{align}
\end{lemma}
\begin{proof}
Observe first that
\begin{align}
\begin{aligned}
    \E[N]
    & = \E\Bigl[N \cdot \llbracket \text{length of } \vec{W} \leq \Rep \rrbracket\Bigr]
    + \E\Bigl[N \cdot \llbracket \text{length of } \vec{W} \geq \Rep+1 \rrbracket\Bigr] \\
    & = \E[N'] + \E\Bigl[N \cdot \llbracket \text{length of } \vec{W} \geq \Rep+1 \rrbracket\Bigr].
\end{aligned}
\end{align}
Conditioned on the event that ``$\text{the length of } \vec{W} \text{ is } k$,''
$\vec{W}$ includes $\frac{|F|}{|E|}k$ edges of $F$ in expectation;
thus, the second term in the above formula can be bounded as
\begin{align}
\begin{aligned}
    & \E\Bigl[N \cdot \llbracket \text{length of } \vec{W} \geq \Rep+1 \rrbracket\Bigr] \\
    & = \Pr\Bigl[\text{length of } \vec{W} \geq \Rep+1\Bigr]
        \cdot \E\Bigl[N \Bigm| \text{length of } \vec{W} \geq \Rep+1\Bigr] \\
    & \leq \Pr\Bigl[\text{length of } \vec{W} \geq \Rep+1\Bigr]
        \cdot \E\Bigl[\text{length of } \vec{W} \Bigm| \text{length of } \vec{W} \geq \Rep+1\Bigr] \cdot \frac{|F|}{|E|} \\
    & \leq \left(1-\frac{1}{\rep}\right)^{100\rep} \cdot (100\rep + \rep) \cdot \frac{|F|}{|E|} \\
    & \leq \rme^{-100} \cdot 101 \rep \cdot \frac{|F|}{|E|}.
\end{aligned}
\end{align}
Consequently, we use \cref{lem:amp:pow:N1} to derive
\begin{align}
    \E[N']
    \geq \rep \frac{|F|}{|E|} \cdot \approxone^2 \cdot (\sqrt{2}-1)^2
    - \rep \frac{|F|}{|E|} \cdot \rme^{-100} \cdot 101
    \geq 0.1715 \cdot \rep \frac{|F|}{|E|}.
\end{align}
completing the proof.
\end{proof}

We now create a \emph{powered constraint graph} $G'$ that emulates the truncated verifier's walks.
One minor but annoying issue is that
each of the verifier's walks is fractionally weighted by its occurrence probability,
although the underlying graph of $G'$ should be unweighted.
We use a weighting scheme by \cite{radhakrishnan2006gap,radhakrishnan2007dinurs}
to circumvent this issue.
Moreover, we show $\frac{\lambda'}{d'} = \bigO\left(\frac{\lambda}{d}\right)$,
where $d'$ is the degree of $G'$ and $\lambda' \coloneq \lambda(G')$.

\begin{lemma}
\label{lem:amp:pow:powered}
    We can construct a powered constraint graph $G' = (V,E',\Sigma'^{d^{\Rep+1}},\Pi')$ in polynomial time
    such that
    for any proof $\psi' \colon V \to \Sigma'^{d^{\Rep+1}}$,
    $\val_{G'}(\psi')$ is at most the probability that the truncated verifier accepts $\psi'$, and
    $\val_{G'}(\psi') = 1$ if the truncated verifier accepts $\psi'$ with probability $1$.
    Moreover,
    $G'$ is a $d'$-regular graph and
    $\frac{\lambda'}{d'} \leq \frac{4^{100}}{100} \cdot \frac{\lambda}{d}$
    \textup{(}if $\lambda \leq \frac{d}{2}$\textup{)},
    where 
    $d' \coloneq (\rep^\Rep - (\rep-1)^\Rep) \cdot d^\Rep$ and
    $\lambda' \coloneq \lambda(G)$.
\end{lemma}
\begin{proof}
We first describe the truncated verifier's walk $\vec{W}$ in the language of adjacency matrices.
Let $\mat{A}$ denote the adjacency matrix of the underlying graph of $G$;
recall that the transition matrix of $G$ is $\frac{1}{d}\mat{A}$.
For each $k \in [\Rep]$,
the verifier's walk (i.e., an ASRW) has length exactly $k$
with probability $(1-p)^{k-1} p$, where $p \coloneq \frac{1}{\rep}$.
With the remaining probability
$1 - \sum_{1 \leq k \leq \Rep} (1-p)^{k-1} p = (1-p)^\Rep$,
the verifier's walk would be truncated.
Conditioned on the event that ``$\text{the length of } \vec{W} \text{ is } k$,''
$\vec{W}$ is equivalent to a length-$k$ random walk starting from a random vertex.
Thus, the adjacency matrix corresponding to all length-$k$ verifier's walks is equal to
\begin{align}
    \Pr\Bigl[\text{length of } \vec{W} = k\Bigr] \cdot \left(\frac{1}{d}\mat{A}\right)^k
    = (1-p)^{k-1}p \cdot d^{-k} \cdot \mat{A}^k.
\end{align}
If the verifier's walk is truncated,
it is considered to be a self-loop;\footnote{
    The respective constraint is defined to be
    $\Sigma'^{d^{\Rep+1}} \times \Sigma'^{d^{\Rep+1}}$,
    which is always satisfied.
}
thus the adjacency matrix corresponding to all such verifier's walks is 
\begin{align}
    \Pr\Bigl[\text{length of } \vec{W} \geq \Rep+1\Bigr] \cdot \mat{I}
    = (1-p)^\Rep \cdot \mat{I},
\end{align}
where $\mat{I}$ is the $V \times V$ identity matrix.
The resulting \emph{fractionally weighted} adjacency matrix
formed by the truncated verifier's walk is equal to
\begin{align}
\label{eq:amp:truncated-fracmat}
    \sum_{1 \leq k \leq \Rep} (1-p)^{k-1} p \cdot d^{-k} \cdot \mat{A}^{k}
    + (1-p)^\Rep \cdot \mat{I}.
\end{align}
Observe that we do not need the diagonal matrix $(1-p)^\Rep \cdot \mat{I}$ because
its removal does not
increase the value of any assignment,
and recall that
each edge associated to a verifier's walk should be \emph{integer weighted}.
Accordingly,
subtracting $(1-p)^\Rep \cdot \mat{I}$ from \cref{eq:amp:truncated-fracmat} and
multiplying \cref{eq:amp:truncated-fracmat} by $\rep^\Rep \cdot d^\Rep$
yields the following adjacency matrix of $G'$ having the desired property:\footnote{
    Obviously, $\mat{A}_{G'}$ is symmetric; i.e., $G'$ is undirected.
}
\begin{align}
    \mat{A}_{G'}
    \coloneq \rep^\Rep \cdot d^\Rep
        \sum_{1 \leq k \leq \Rep} (1-p)^{k-1} p \cdot d^{-k} \cdot \mat{A}^{k}
    = \sum_{1 \leq k \leq \Rep} (\rep-1)^{k-1} \rep^{\Rep-k} \cdot d^{\Rep-k} \cdot \mat{A}^{k}.
\end{align}

With this representation in mind,
we construct a powered constraint graph
$G' = (V,E',\Sigma'^{d^{\Rep + 1}}, \allowbreak \Pi' = (\pi'_{e'})_{e' \in E'})$
emulating the truncated verifier's walk as follows:

\begin{itembox}[l]{\textbf{Construction of a powered constraint graph $G'$.}}
\begin{algorithmic}[1]
    \State enumerate all $k$-length walks in $G$ for each $k \in [\Rep]$.
    \State let $G'$ be an edgeless constraint graph on the vertex set $V$.
    \For{\textbf{each} $k$-length walk $\vec{W}$, starting from $x$ and ending at $y$\footnote{
        The terminals $x$ and $y$ need not be distinguished.
    }}
        \State create a fresh edge $(x,y)$ and add it into $E(G')$.
        \State represent
        by $\pi'_{(x,y)} \subseteq \Sigma'^{d^{\Rep+1}} \times \Sigma'^{d^{\Rep+1}}$
        (the conjunction of) the tests performed by the verifier along $\vec{W}$.
        \State make $(\rep-1)^{k-1} \rep^{\Rep-k} \cdot d^{\Rep-k} $ number of its copies.
    \EndFor
\end{algorithmic}
\end{itembox}
The construction can be done in polynomial time in the size of $G$.
Observe that $G'$ is regular; indeed,
it is $(\rep^\Rep - (\rep-1)^\Rep) \cdot d^\Rep$-regular because we have
\begin{align}
\begin{aligned}
    \mat{A}_{G'} \vec{1}
    & = \sum_{1 \leq k \leq \Rep} (\rep-1)^{k-1} \rep^{\Rep-k} \cdot d^{\Rep-k} \cdot (d^{k} \vec{1}) \\
    & = \rep^{\Rep-1} d^\Rep \cdot \vec{1} \sum_{0 \leq k \leq \Rep-1} \left(\frac{\rep-1}{\rep}\right)^k \\
    & = (\rep^\Rep - (\rep-1)^\Rep) \cdot d^\Rep \cdot \vec{1}.
\end{aligned}
\end{align}
Let $d' \coloneq (\rep^\Rep - (\rep-1)^\Rep) \cdot d^\Rep$ be the degree of $G'$ and
$\lambda' \coloneq \lambda(G')$.
By \cref{lem:lambda}, we have
\begin{align}
\begin{aligned}
    \lambda'
    & \leq \sum_{1 \leq k \leq \Rep} (\rep-1)^{k-1} \rep^{\Rep-k} \cdot d^{\Rep-k} \cdot \lambda^k \\
    & \leq \rep^{\Rep-1} \cdot d^\Rep \sum_{1 \leq k \leq \Rep} \left(\frac{\lambda}{d}\right)^k \\
    & = \rep^{\Rep-1} \cdot d^\Rep \frac{\lambda}{d} \frac{1-\left(\frac{\lambda}{d}\right)^\Rep}{1-\frac{\lambda}{d}} \\
    & = \rep^{\Rep-1} \frac{d^\Rep - \lambda^\Rep}{d-\lambda} \lambda \\
    & \leq 2 \cdot \rep^{\Rep-1} \cdot d^{\Rep-1} \cdot \lambda,
\end{aligned}
\end{align}
where the last inequality used the assumption that $\lambda \leq \frac{d}{2}$.
Simple calculation using the inequality
$d' \geq \Rep(\rep-1)^{\Rep-1} \cdot d^\Rep$
yields that
\begin{align}
    \frac{\lambda'}{d'}
    \leq \frac{2 \cdot \rep^{\Rep-1}\cdot d^{\Rep-1} \cdot\lambda}{\Rep(\rep-1)^{\Rep-1} \cdot d^{\Rep}}
    = \frac{2(\rep-1)}{100 \rep \cdot \rep}\left(\frac{\rep}{\rep-1}\right)^{100 \rep} \cdot \frac{\lambda}{d}
    \leq \frac{4^{100}}{100} \cdot \frac{\lambda}{d},
\end{align}
where the last inequality holds because
$\left(\frac{\rep}{\rep-1}\right)^\rep \leq 4$ as $\rep \geq 2$,
completing the proof.
\end{proof}

We eventually accomplish the proof of \cref{lem:amp:pow}.

\begin{proof}[Proof of \cref{lem:amp:pow}]
The completeness (i.e., $\val_G(\psi^\ini \reco \psi^\tar)=1$ implies $\val_{G'}(\psi'^\ini \reco \psi'^\tar)=1$)
follows from \cref{lem:amp:pow:complete},
which still holds
on the truncated verifier's walk and thus on $G'$.
Then, we show the soundness, i.e.,
$\val_G(\psi^\ini \reco \psi^\tar) < 1-\epsilon$ implies
$\val_{G'}(\psi'^\ini \reco \psi'^\tar) < 1-\epsilon'$.
Recall that
$\psi \coloneq \psi^{(\ttt)}$ is a ``bad'' assignment such that $\val_G(\psi^{(\ttt)}) < 1-\delta$, and
let $\psi' \coloneq \psi'^{(\ttt)}$.
Using \cref{lem:amp:pow:N2,lem:amp:pow:truncated-N1} and
the {Paley}--{Zygmund} inequality \cite[Section~4.8, Ex.~1]{alon2016probabilistic},
we derive
\begin{align}
\begin{aligned}
    \Pr\Bigl[\text{truncated verifier rejects } \psi'\Bigr]
    & \geq \Pr[N' > 0] \\
    & \geq \frac{\E[N']^2}{\E[N'^2]} \\
    & \geq \rep \frac{|F|}{|E|} \cdot \frac{0.1715^2}{3 + 2\frac{d}{d-\lambda}} \\
    & \underbrace{\geq}_{\text{\cref{clm:amp:FE}}} \frac{0.1715^2}{3 + 2\frac{d}{d-\lambda}}
        \cdot \min\left\{\rep \delta, 1-\frac{\rep}{|E|}\right\}.
\end{aligned}
\end{align}
Hereafter, 
we can safely assume that $|V| \geq \frac{2d^{\Rep+2}}{\epsilon}$.\footnote{
    Otherwise, we can
    solve \prb{Gap$_{1,1-\epsilon}$ \BCSPReconf[$_W$]}
    in constant time.
}
By \cref{clm:amp:interpolate}, we have
\begin{align}
    \rep \delta
    \geq \left\lceil \frac{2}{\epsilon} \right\rceil \left(\epsilon - \frac{d^{\Rep+2}}{|V|}\right)
    \geq \left\lceil \frac{2}{\epsilon} \right\rceil \frac{\epsilon}{2}
    \geq 1.
\end{align}
By using that
$|E| \geq |V| \geq \frac{2^{102}}{\epsilon}$, we obtain
\begin{align}
    1-\frac{\rep}{|E|}
    \geq 1- \left\lceil \frac{2}{\epsilon} \right\rceil \frac{1}{|E|}
    \geq 1 - \left\lceil \frac{2}{\epsilon} \right\rceil \frac{\epsilon}{2^{102}}
    \geq 1-\frac{1}{2^{100}},
\end{align}
where the last inequality is because
$\left\lceil \frac{2}{\epsilon} \right\rceil \leq \frac{4}{\epsilon}$
for any $\epsilon \in (0,1)$.
Consequently, by \cref{lem:amp:pow:powered}, we derive
\begin{align}
\begin{aligned}
    1-\val_{G'}(\psi'^\ini \reco \psi'^\tar)
    & \underbrace{\geq}_{\text{\cref{lem:amp:pow:powered}}}
        \Pr\Bigl[\text{truncated verifier rejects } \psi'\Bigr] \\
    & \geq \frac{0.1715^2}{3 + 2\frac{d}{d-\lambda}} \cdot \min\left\{1, 1-\frac{1}{2^{100}}\right\} \\
    & > \underbrace{\frac{0.0294}{3 + 2\frac{d}{d-\lambda}}}_{=\epsilon'},
\end{aligned}
\end{align}
thereby finishing the soundness.
An upper bound of $\frac{\lambda'}{d'}$ follows from \cref{lem:amp:pow:powered}.
Observe that the alphabet size of $G'$ is equal to
\begin{align}
    \left|\Sigma'^{d^{\Rep + 1}}\right| = \left(\frac{W(W+1)}{2}\right)^{d^{100\left\lceil\frac{2}{\epsilon}\right\rceil+1}},
\end{align}
which accomplishes the proof.
\end{proof}

\subsection{Putting It Together: Proof of \cref{thm:amp}}
\label{subsec:amp:together}
We are now ready to accomplish the proof of \cref{thm:amp}
using \cref{lem:amp:expand,lem:amp:pow}.

\begin{proof}[Proof of \cref{thm:amp}]
Let $\rho \in (0,1)$ be a given parameter.
Under RIH,
\prb{Gap$_{1,1-\epsilon}$ \BCSPReconf[$_W$]}
is \PSPACE-hard
even if the underlying graph is $\Delta$-regular
for some $\epsilon \in (0,1)$ and $q, W, \Delta \in \bbN$ \cite{ohsaka2023gap}.
Using \cref{lem:amp:expand}, for any positive even integer $d_0 \geq 3$,
we can reduce \prb{Gap$_{1,1-\epsilon}$ \BCSPReconf[$_W$]} to 
\prb{Gap$_{1,1-\frac{\Delta}{\Delta+d_0} \epsilon}$ \BCSPReconf[$_W$]}
whose underlying graph is $(d,\lambda)$-expander,
where $d = \Delta + d_0$ and
$\lambda = \Delta + 2\sqrt{d_0}$.
Setting $d_0 \coloneq C \cdot \Delta$, where $C$ is a positive even integer defined as
\begin{align}
    C \coloneq 2\cdot \left\lceil \left( \frac{3 \cdot 4^{100}}{100 \cdot \rho} \right)^2 \right\rceil,
\end{align}
we have
\begin{align}
    \frac{\lambda}{d}
    = \frac{\Delta + 2\sqrt{C \cdot \Delta}}{\Delta + C \cdot \Delta}
    \leq \frac{2\sqrt{C} + 1}{C + 1}
    \leq \frac{3\sqrt{C}}{C}
    \leq \frac{3}{\sqrt{C}}
    \leq \frac{100}{4^{100}} \cdot \rho.
\end{align}
Using \cref{lem:amp:pow},
we can further reduce
\prb{Gap$_{1,1-\frac{\Delta}{\Delta+d_0} \epsilon}$ \BCSPReconf[$_W$]}
to
\prb{Gap$_{1,1-\epsilon'}$ \BCSPReconf[$_{W'}$]}
whose underlying graph is $(d',\lambda')$-expander, where
\begin{align}
    \epsilon' & = \frac{0.0294}{3 + 2\frac{d}{d-\lambda}} \geq 0.0058, \\
    \frac{\lambda'}{d'} & \leq \frac{4^{100}}{100} \cdot \frac{\lambda}{d} \leq \rho,
\end{align}
and $W'$ and $d'$ are dependent only on the values of
$\epsilon$, $W$, and $d$.
A couple of these gap-preserving reductions accomplishes the proof.
\end{proof}

\section{\PSPACE-hardness of Approximation for \MinmaxSetCoverReconf}
\label{sec:sc}
As an application of \cref{thm:amp},
we show \PSPACE-hardness of approximation for \MinmaxSetCoverReconf under RIH.
Let $\calU$ be a finite set called the \emph{universe} and
$\calF = \{ S_1, \ldots, S_m \}$ be a family of $m$ subsets of $\calU$.
For a set system $(\calU,\calF)$, a \emph{cover} is
a subfamily of $\calF$ whose union is equal to $\calU$.
For two covers $\calC^\ini$ and $\calC^\tar$ for $(\calU,\calF)$,
a \emph{reconfiguration sequence from $\calC^\ini$ to $\calC^\tar$}
is any sequence $( \calC^{(1)}, \ldots, \calC^{(T)})$ of covers such that
$\calC^{(1)} = \calC^\ini$,
$\calC^{(T)} = \calC^\tar$, and
each $\calC^{(\ttt+1)}$ is obtained from $\calC^{(\ttt)}$
by adding or removing a single set of $\calF$
(i.e., $|\calC^{(\ttt)} \triangle \calC^{(\ttt+1)}| \leq 1$).
The \SetCoverReconf problem \cite{ito2011complexity} asks to
decide for a set system $(\calU, \calF)$ and its two covers $\calC^\ini$ and $\calC^\tar$,
if there is a reconfiguration sequence from $\calC^\ini$ to $\calC^\tar$
consisting of covers of size at most $k+1$.
\SetCoverReconf is \PSPACE-complete \cite{ito2011complexity}.
Subsequently, we formulate its optimization version.
Denote by $\OPT(\calF)$ the size of the minimum cover of $(\calU,\calF)$.
For a reconfiguration sequence
$\scrC = ( \calC^{(1)}, \ldots, \calC^{(\TTT)} )$
from $\calC^\ini$ to $\calC^\tar$, 
its \emph{cost} is defined as the maximum value of 
$\frac{|\calC^{(\ttt)}|}{\OPT(\calF)+1}$ over all $\calC^{(\ttt)}$'s in $\scrC$; namely,
\begin{align}
    \cost_{\calF}(\scrC)
    \coloneq \max_{\calC^{(\ttt)} \in \scrC} \frac{|\calC^{(\ttt)}|}{\OPT(\calF)+1},
\end{align}
where division by $\OPT(\calF)+1$ comes from the nature that we must first add at least one set whenever $|\calC^\ini| = |\calC^\tar| = \OPT(\calF)$ and $\calC^\ini \neq \calC^\tar$.
In \MinmaxSetCoverReconf,
we wish to minimize $\cost_\calF(\scrC)$ subject to
$\scrC = ( \calC^\ini, \ldots, \calC^\tar )$.
For a pair of covers $\calC^\ini$ and $\calC^\tar$,
let
$\cost_{\calF}(\calC^\ini \reco \calC^\tar)$
denote the minimum value of $\cost_\calF(\scrC)$
over all possible reconfiguration sequences $\scrC$ from $\calC^\ini$ to $\calC^\tar$; namely,
\begin{align}
    \cost_{\calF}(\calC^\ini \reco \calC^\tar)
    \coloneq \min_{\scrC = ( \calC^\ini, \ldots, \calC^\tar )} \cost_{\calF}(\scrC).
\end{align}
For every $1 \leq c \leq s$,
\prb{Gap$_{c,s}$ \SetCoverReconf} requests
to distinguish whether
$\cost_{\calF}(\calC^\ini \reco \calC^\tar) \leq c$ or
$\cost_{\calF}(\calC^\ini \reco \calC^\tar) > s$.
The result in this section is
the gap reducibility from
\MaxminBCSPReconf to
\MinmaxSetCoverReconf shown below.

\begin{theorem}
\label{thm:sc}
    For every number $\epsilon \in (0,1)$,
    positive integers $W, d \in \bbN$, and
    number $\lambda > 0$
    such that $\frac{\lambda}{d} \leq \frac{\epsilon}{3}$,
    there exists a gap-preserving reduction from
    \prb{Gap$_{1,\epsilon'}$ \BCSPReconf[$_W$]}
    whose underlying graph is $(d,\lambda)$-expander
    to
    \prb{Gap$_{1,2-\epsilon}$ \SetCoverReconf},
    where 
    \begin{align}
        \epsilon' \coloneq \epsilon^2 \cdot \frac{d}{d+\lambda} \cdot
            \left(1-\frac{3 \lambda}{\epsilon d} \right).
    \end{align}
\end{theorem}\noindent
The value of $\epsilon'$ can be made arbitrarily close to $\epsilon^2$
by making $\frac{\lambda}{d}$ sufficiently small.
The following corollaries are thus immediate from \cref{thm:amp,thm:sc}
(notice: $(2-1.0029)^2 \approx 0.9942$).

\begin{corollary}
\label{cor:sc}
    Under RIH, \prb{Gap$_{1,1.0029}$ \SetCoverReconf} is \PSPACE-hard.
\end{corollary}
\begin{proof}
Let $\epsilon \coloneq 0.9971$.
By \cref{thm:amp},
\prb{Gap$_{1,0.9942}$ \BCSPReconf[$_W$]}
over $(d,\lambda)$-expander graphs with
$\frac{\lambda}{d} \leq 10^{-10} \leq \frac{\epsilon}{3}$
is \PSPACE-hard under RIH.
Since it holds that
\begin{align}
    \epsilon^2 \cdot \frac{d}{d+\lambda} \cdot \left( 1-\frac{3\lambda}{\epsilon d} \right)
    > 0.9942,
\end{align}
we can use \cref{thm:sc} to further reduce it to
\prb{Gap$_{1,2-\epsilon}$ \SetCoverReconf},
where $2-\epsilon = 1.0029$, thereby completing the proof.
\end{proof}

\begin{corollary}
\label{cor:ds}
Under RIH, \prb{Gap$_{1,1.0029}$ \DominatingSetReconf} is \PSPACE-hard.
\end{corollary}
\begin{proof}
Note that \DominatingSetReconf \cite{haddadan2016complexity,suzuki2016reconfiguration}
and its optimization version are defined analogously to \SetCoverReconf (see also~\cite{ito2011complexity,ohsaka2023gap}).
It suffices to prove that for every number $\epsilon \in (0,1)$,
there exists a gap-preserving reduction from
\prb{Gap$_{1,1+\epsilon}$ \SetCoverReconf}
to
\prb{Gap$_{1,1+\epsilon}$ \DominatingSetReconf}.

The proof uses a Karp reduction from \prb{Set Cover} to \prb{Dominating Set}
due to \citet{paz1981non}.
Let $(\calC^\ini, \calC^\tar; \calU, \calF)$
be an instance of \MinmaxSetCoverReconf,
where $\calF = \{S_1, \ldots, S_m\}$ is a family of $m$ subsets of a universe $\calU$, and
$\calC^\ini$ and $\calC^\tar$ are minimum covers for $(\calU, \calF)$.
We first construct a graph $G = (V,E)$, where
\begin{align}
    V & \coloneq \calF \cup \calU, \\
    E & \coloneq {\calF \choose 2} \cup
    \Bigl\{ (e, S_j) \in \calU \times \calF \Bigm| e \in S_j \Bigr\}.
\end{align}
Here, $G$ is a split graph with a clique on $\calF$ and an independent set on $\calU$.
Denote by $\gamma(G)$ the minimum size of any dominating set of $G$.
Observe easily that any (minimum) cover for $(\calU,\calF)$
is a (minimum) dominating set of $G$; in particular, $\gamma(G) = \OPT(\calF)$.
Constructing $D^\ini \coloneq \calC^\ini$ and $D^\tar \coloneq \calC^\tar$,
we obtain an instance $(D^\ini, D^\tar; G)$ of \MinmaxDominatingSetReconf,
thereby completing the description of the reduction.

The completeness (i.e.,
$\cost_\calF(\calC^\ini \reco \calC^\tar) = 1$ implies
$\cost_G(D^\ini \reco D^\tar) = 1$)
is immediate because any reconfiguration sequence of covers from $\calC^\ini$ to $\calC^\tar$
is a reconfiguration sequence of dominating sets from $D^\ini$ to $D^\tar$.
We then show the soundness; i.e.,
$\cost_\calF(\calC^\ini \reco \calC^\tar) > 1+\epsilon$ implies
$\cost_G(D^\ini \reco D^\tar) > 1+\epsilon$.
Let $\scrD = ( D^{(1)}, \ldots, D^{(\TTT)} )$
be any reconfiguration sequence of dominating sets from $D^\ini$ to $D^\tar$ whose value is equal to
$\cost_G(D^\ini \reco D^\tar)$.
Construct then a sequence of covers,
$\scrC = ( \calC^{(1)}, \ldots, \calC^{(\TTT)} )$, where each $\calC^{(\ttt)}$ is defined as
\begin{align}
    \calC^{(\ttt)}
    \coloneq \Bigl(D^{(\ttt)} \cap \calF\Bigr)
    \cup \Bigl\{ S_j \in \calF \Bigm| \exists e \in D^{(\ttt)} \text{ s.t. } e \in S_j \Bigr\}.
\end{align}
Since $\scrC$ is a valid reconfiguration sequence from $\calC^\ini$ to $\calC^\tar$,
it holds that $\cost_\calF(\scrC) > 1+\epsilon$; in particular,
there exists some $\calC^{(\ttt)}$ such that $|\calC^{(\ttt)}| > (1+\epsilon)(\OPT(\calF)+1)$.
Since $|\calC^{(\ttt)}| \leq |D^{(\ttt)}|$, we obtain
\begin{align}
    \cost_G(D^\ini \reco D^\tar) = \cost_G(\scrD)
    \geq \frac{|D^{(\ttt)}|}{\gamma(G) + 1}
    > \frac{(1+\epsilon)(\OPT(\calF) + 1)}{\gamma(G) + 1}
    = 1+\epsilon,
\end{align}
as desired.
\end{proof}

\paragraph{Reduction.}
In the remainder of this section, we prove \cref{thm:sc}.
Our gap-preserving reduction from
\MaxminBCSPReconf to \MinmaxSetCoverReconf
is based on
that from \prb{Label Cover} to \prb{Set Cover} due to
\citet{lund1994hardness}.
We first introduce a helpful gadget for representing a binary constraint (e.g., \cite{gupta2008course}).
Let $\Sigma$ be a finite alphabet.
For each value $\alpha \in \Sigma$ and each subset $S \subseteq \Sigma$,
we define
$\bar{Q_\alpha} \subset \{0,1\}^\Sigma$ and
$Q_{S} \subset \{0,1\}^\Sigma$
as follows:
\begin{align}
    \bar{Q_\alpha}
    & \coloneq \Bigl\{ \vec{q} \in \{0,1\}^\Sigma \Bigm| q_\alpha = 0 \Bigr\}, \\
    Q_{S}
    & \coloneq \Bigl\{ \vec{q} \in \{0,1\}^\Sigma \Bigm| q_\alpha = 1 \text{ for some } \alpha \in S \Bigr\}.
\end{align}
A collection of $\bar{Q_\alpha}$'s and $Q_{S}$'s has the following property.

\begin{observation}
\label{obs:sc:Q}
    Let $\pi \subseteq \Sigma^2$ represent a binary constraint.
    For each value $\alpha, \beta \in \Sigma$, define
        $S_\alpha \coloneq \bar{Q_\alpha}$ and
        $T_\beta \coloneq Q_{\pi(\beta)}$,
        where
        $\pi(\beta) \coloneq \{ \alpha \in \Sigma \mid (\alpha,\beta) \in \pi \}$.
    Then, any subfamily $\calC$ of
    $\{S_\alpha \mid \alpha \in \Sigma\} \cup \{T_\beta \mid \beta \in \Sigma\}$
    covers $\{0,1\}^\Sigma$
    if and only if
    $\calC$ includes $S_\alpha$ and $T_\beta$ such that $(\alpha, \beta) \in \pi$.
\end{observation}
\begin{proof}
Since the ``if'' direction is obvious, we prove (the contraposition of) the ``only-if'' direction.
Suppose that $ (\alpha, \beta) \notin \pi$ for every $S_\alpha, T_\beta \in \calC$.
Then, we define a vector $\vec{q}^* \in \{0,1\}^\Sigma$ as follows:
\begin{align}
    q^*_\alpha \coloneq
    \begin{cases}
        1 & \text{if } S_\alpha \in \calC \\
        0 & \text{otherwise}
    \end{cases}
    \text{  for all } \alpha \in \Sigma.
\end{align}
We claim that $\calC$ does not cover $\vec{q}^*$.
On one hand, it holds that $\vec{q}^* \notin S_\alpha$ whenever $S_\alpha \in \calC$ since $q^*_\alpha = 1$.
On the other hand, if $\vec{q}^* \in T_\beta \in \calC$,
we have $q^*_\alpha = 1$ for some $\alpha \in \pi(\beta)$.
However, by definition of $\vec{q}^*$, $\calC$ must include $S_\alpha$,
implying that $S_\alpha, T_\beta \in \calC$ such that $(\alpha, \beta) \in \pi$,
a contradiction.
\end{proof}

We now describe the reduction.
Let $(\psi^\ini,\psi^\tar; G)$ be an instance of 
\prb{Gap$_{1,\epsilon'}$ \BCSPReconf[$_W$]}, where
$\epsilon' = \epsilon^2 \cdot \frac{d}{d+\lambda} \cdot \left(1-\frac{3 \lambda}{\epsilon d} \right)$
for $\epsilon \in (0,1)$, and
$G = (V,E,\Sigma,\Pi)$ is a constraint graph such that
$(V,E)$ is $(d,\lambda)$-expander such that $\frac{\lambda}{d} \leq \frac{\epsilon}{3}$,
$|\Sigma| = W$, and
$\psi^\ini$ and $\psi^\tar$ satisfy $G$.
Hereafter, we assume to be given an arbitrary order $\prec$ over $V$.
Create an instance of \MinmaxSetCoverReconf as follows.
Define $B \coloneq \{0,1\}^\Sigma$.
For each vertex $v \in V$ and each value $\alpha \in \Sigma$,
we define $S_{v,\alpha} \subset E \times B$ as
\begin{align}
    S_{v,\alpha} \coloneq
    \left( \bigcup_{e=(v,w) \in E : v \prec w} \{e\} \times \bar{Q_\alpha} \right)
    \cup
    \left( \bigcup_{e=(v,w) \in E : v \succ w} \{e\} \times Q_{\pi_e(\alpha)} \right),
\end{align}
where we define
\begin{align}
    \pi_e(\alpha) \coloneq \Bigl\{
        \beta \in \Sigma \Bigm| (\alpha,\beta) \in \pi_e
    \Bigr\}.
\end{align}
Then, a set system $(\calU, \calF)$ is defined as
\begin{align}
    \calU & \coloneq E \times B, \\
    \calF & \coloneq \Bigl\{ S_{v,\alpha} \Bigm| v \in V, \alpha \in \Sigma \Bigr\}.
\end{align}
Given a satisfying assignment $\psi \colon V \to \Sigma$ for $G$,
we associate it with a subfamily $\calC_\psi \subset \calF$ such that
\begin{align}
\label{eq:sc:Cpsi}
    \calC_\psi \coloneq \Bigl\{ S_{v,\psi(v)} \Bigm| v \in V \Bigr\},
\end{align}
which turns out to be a cover for $(\calU, \calF)$ because of \cref{obs:sc:Q}.
Since any cover for $(\calU, \calF)$ must include some
$S_{v,\alpha}$ for each $v \in V$
(as will be shown in \cref{clm:sc:Lv}),
$\calC_\psi$ must be a minimum cover;
i.e., $|\calC_\psi| = |V| = \OPT(\calF)$.
Constructing minimum covers $\calC^\ini$ from $\psi^\ini$ and $\calC^\tar$ from $\psi^\tar$
according to \cref{eq:sc:Cpsi},
we obtain an instance
$(\calC^\ini, \calC^\tar; \calU,\calF)$
of \MinmaxSetCoverReconf, finishing the description of the reduction.

\paragraph{Correctness.}
The completeness is immediate from the construction.
\begin{lemma}
\label{lem:sc:complete}
    If
    $\val_G(\psi^\ini \reco \psi^\tar) = 1$,
    then
    $\cost_\calF(\calC^\ini \reco \calC^\tar) \leq 1$.
\end{lemma}
\begin{proof}
It suffices to consider the case that $\psi^\ini$ and $\psi^\tar$ differ in exactly one vertex, say, $v$.
Since $|\calC^\ini \cup \calC^\tar| = |\calC^\ini|+1$,
a sequence
$\scrC = ( \calC^\ini, \calC^\ini \cup \calC^\tar, \calC^\tar )$
is a valid reconfiguration sequence from $\calC^\ini$ to $\calC^\tar$;
we thus have
\begin{align}
    \cost_\calF(\scrC)
    = \frac{\max\Bigl\{
        |\calC^\ini|,
        |\calC^\ini \cup \calC^\tar|,
        |\calC^\tar|
    \Bigr\}}{\OPT(\calF)+1}
    = 1,
\end{align}
completing the proof.
\end{proof}

Subsequently, we would like to prove the soundness.
\begin{lemma}
\label{lem:sc:soundness}
If $|V| \geq \frac{d^2}{\lambda^2}$ and
$\val_G(\psi^\ini \reco \psi^\tar) <
    \epsilon'
$, then
$\cost_\calF(\calC^\ini \reco \calC^\tar) > 2 - \epsilon$.
\end{lemma}
Hereafter,
we can safely assume that $|V| \geq \frac{d^2}{\lambda^2}$ because otherwise,
we can solve \prb{Gap$_{1,\epsilon'}$ \BCSPReconf[$_W$]} in constant time.

Suppose we are given a reconfiguration sequence
$\scrC = ( \calC^{(1)}, \ldots, \calC^{(\TTT)} )$
from $\calC^\ini$ to $\calC^\tar$
such that $\cost_\calF(\scrC) \leq 2-\epsilon$.
For each vertex $v\in V$ and each integer $\ttt \in [\TTT]$,
we define $L_v^{(\ttt)}$ as
the set of $v$'s values assigned by $\calC^{(\ttt)}$; namely,
\begin{align}
\label{eq:sc:Lv}
    L_v^{(\ttt)} \coloneq \Bigl\{ \alpha \in \Sigma \Bigm| S_{v,\alpha} \in \calC^{(\ttt)} \Bigr\}.
\end{align}
Note that
\begin{align}
    \sum_{v \in V} |L_v^{(\ttt)}| = |\calC^{(\ttt)}|.
\end{align}
We claim that each edge $(v,w) \in E$ is satisfied by 
some $(\alpha,\beta) \in L_v^{(\ttt)} \times L_w^{(\ttt)}$.

\begin{claim}
\label{clm:sc:Lv}
    Suppose that a subfamily $\calC \subseteq \calF$ covers $\calU$, and
    $L_v$ for each $v \in V$ is defined by \cref{eq:sc:Lv}.
    Then, for each edge $e = (v,w) \in E$,
    there exists a pair of values $\alpha \in L_v$ and $\beta \in L_w$ such that
    $(\alpha, \beta) \in \pi_e$\textup{;}
    i.e., $e$ is satisfied by $(\alpha, \beta)$.
    In particular, $|L_v| \geq 1$ for all $v \in V$.
\end{claim}
\begin{proof}
For an edge $e=(v,w) \in E$,
assume that $v \prec w$ without loss of generality.
Observe that $\calC$ particularly covers $\{e\} \times B$.
Since any $S_{x,\alpha}$ with $x \notin \{v,w\}$ is disjoint to $\{e\} \times B$,
\begin{align}
    \left(\bigcup_{\alpha \in L_v} S_{v,\alpha}\right)
    \cup \left(\bigcup_{\beta \in L_w} S_{w,\beta}\right)
    \subset \calC
\end{align}
must cover $\{e\} \times B$.
In particular, it holds that
\begin{align}
\begin{aligned}
    \{e\} \times B
    & = \left[
        \left(
            \bigcup_{\alpha \in L_v} S_{v,\alpha}
        \right) \cup \left(
            \bigcup_{\beta \in L_w} S_{w,\beta}
        \right)
        \right] \cap (\{e\} \times B) \\
    & = \left(
        \bigcup_{\alpha \in L_v} S_{v,\alpha} \cap \Bigl(\{e\} \times B\Bigr)
        \right) \cup
        \left(
        \bigcup_{\beta \in L_w} S_{w,\beta} \cap \Bigl(\{e\} \times B\Bigr)
        \right) \\
    & = \{e\} \times \left( \bigcup_{\alpha \in L_v} \bar{Q_\alpha} \cup \bigcup_{\beta \in L_w} Q_{\pi_e(\beta)} \right).
\end{aligned}
\end{align}
Consequently, we have
\begin{align}
    \bigcup_{\alpha \in L_v} \bar{Q_\alpha} \cup \bigcup_{\beta \in L_w} Q_{\pi_e(\beta)} = B.
\end{align}
By \cref{obs:sc:Q}, there must be a pair of $\alpha \in L_v$ and $\beta \in L_w$ such that
$(\alpha, \beta) \in \pi_e$, as desired.
\end{proof}

Then, we extract from $L_v^{(\ttt)}$'s an assignment
$\psi^{(\ttt)} \colon V \to \Sigma$ as follows:
\begin{align}
    \psi^{(\ttt)}(v) \coloneq 
    \begin{cases}
        \text{unique } \alpha \text{ in } L_v^{(\ttt)} & \text{if } |L_v^{(\ttt)}| = 1, \\
        \aaa & \text{otherwise},
    \end{cases}
\end{align}
where $\aaa$ is any default value of $\Sigma$.
Observe that
$\psi^{(1)} = \psi^\ini$ and $\psi^{(\TTT)} = \psi^\tar$.
By assumption, $\calC^{(\ttt)}$ and $\calC^{(\ttt+1)}$ differ in exactly one set,
implying that $L_v^{(\ttt)} \neq L_v^{(\ttt+1)}$ for at most one vertex $v$; i.e.,
$\psi^{(\ttt)}$ and $\psi^{(\ttt+1)}$ differ in at most one vertex.
Therefore, $( \psi^{(1)}, \ldots, \psi^{(\TTT)} )$
is a valid reconfiguration sequence from $\psi^\ini$ to $\psi^\tar$.
Since edge $e=(v,w)$ is satisfied by $\psi^{(\ttt)}$ whenever
$|L_v^{(\ttt)}| = |L_w^{(\ttt)}| = 1$ owing to \cref{clm:sc:Lv},
we wish to estimate the number of such edges.
We define $S^{(\ttt)}$ as the set of vertices $v$ such that $L_v^{(\ttt)}$ is a singleton; namely,
\begin{align}
    S^{(\ttt)} \coloneq \Bigl\{ v \in V \Bigm| |L_v^{(\ttt)}| = 1 \Bigr\}.
\end{align}
Below, we claim that $|S^{(\ttt)}|$ is approximately greater than $\epsilon |V|$.

\begin{claim}
\label{clm:sc:Si}
    If $|V| \geq \frac{d^2}{\lambda^2}$,  then
    $|S^{(\ttt)}| \geq \left(1-\frac{\lambda}{d}\right) \cdot \epsilon |V|$
    for all $\ttt \in [\TTT]$.
\end{claim}
\begin{proof}
Observe first that
$|S^{(\ttt)}| \geq \epsilon |V| - (2-\epsilon)$;
because otherwise, \cref{clm:sc:Lv} implies
\begin{align}
\begin{aligned}
    |\calC^{(\ttt)}|
    & = \sum_{v \in V} |L_v^{(\ttt)}| \\
    & \geq \underbrace{1 \cdot |S^{(\ttt)}|}_{\text{contribution of } |L_v^{(\ttt)}|=1}
        + \underbrace{2 \cdot (|V| - |S^{(\ttt)}|)}_{\text{contribution of } |L_v^{(\ttt)}| \geq 2} \\
    & = 2|V| - |S^{(\ttt)}| \\
    & > (2 - \epsilon ) \cdot (|V| + 1),
\end{aligned}
\end{align}
contradicting the assumption that $\cost_{\calF}(\calC^{(\ttt)}) \leq 2-\epsilon$.
Since 
$\sqrt{|V|} \geq \frac{d}{\lambda} \geq \frac{3}{\epsilon} $ by assumption, it holds that
\begin{align}
\begin{aligned}
    |S^{(\ttt)}|
    \geq \epsilon|V| - (2-\epsilon)
    \geq \left(1-\frac{2}{\epsilon|V|}\right) \cdot \epsilon |V|
    \geq \left(1-\frac{1}{\sqrt{|V|}}\right) \cdot  \epsilon |V|
    \geq \left(1-\frac{\lambda}{d}\right) \cdot \epsilon |V|,
\end{aligned}
\end{align}
as desired.
\end{proof}

The proof of \cref{lem:sc:soundness} and thus \cref{thm:sc} is obtained by
applying \cref{clm:sc:Lv,clm:sc:Si} and
the expander mixing lemma \cite{alon1988explicit},
which states that
$e_G(S,T)$ of an expander graph $G$ is concentrated around
its expectation if $G$ were a \emph{random} $d$-regular graph.
\begin{lemma}[Expander mixing lemma \cite{alon1988explicit}]
\label{lem:expander-mixing}
    Let $G$ be a $(d,\lambda)$-expander graph over $n$ vertices.
    Then, for any two sets $S$ and $T$ of vertices,
    it holds that
    \begin{align}
        \left|e_G(S,T) -  \frac{d|S| \cdot |T|}{n} \right| \leq \lambda \sqrt{|S|\cdot|T|}.
    \end{align}
\end{lemma}

\begin{proof}[Proof of \cref{lem:sc:soundness}]
By applying \cref{clm:sc:Si} and the expander mixing lemma \cite{alon1988explicit}
on $S^{(\ttt)}$ for $\ttt \in [\TTT]$,
we derive
\begin{align}
\begin{aligned}
    e_G(S^{(\ttt)},S^{(\ttt)})
    & \geq \frac{d \cdot |S^{(\ttt)}|^2}{|V|} - \lambda \cdot |S^{(\ttt)}| \\
    & \overset{\heartsuit}{\geq} \frac{d \cdot \epsilon^2 \left(1-\frac{\lambda}{d}\right)^2 |V|^2}{|V|}
        - \lambda \cdot \epsilon \left(1-\frac{\lambda}{d}\right) |V| \\
    & = \epsilon^2 \cdot \left(1-\frac{\lambda}{d}\right) \cdot
        \left(1-\frac{\lambda}{d}-\frac{\lambda}{\epsilon d}\right) \cdot d|V| \\
    & \underbrace{\geq}_{\text{by } \epsilon \in (0,1)} \epsilon^2 \cdot \left(1-\frac{\lambda}{\epsilon d}\right) \cdot
        \left(1-\frac{2\lambda}{\epsilon d}\right) \cdot d|V| \\
    & \geq \epsilon^2 \cdot \left(1-\frac{3 \lambda}{\epsilon d} \right) \cdot d|V|,
\end{aligned}
\end{align}
where the second inequality ($\heartsuit$) holds for the following reason:
$\frac{d \cdot |S^{(\ttt)}|^2}{|V|} - \lambda \cdot |S^{(\ttt)}|$
as a quadratic polynomial in $|S^{(\ttt)}|$ is monotonically increasing
when $|S^{(\ttt)}| \geq \frac{\lambda}{2d} |V|$;
besides, we have
$|S^{(\ttt)}| \geq \left(1-\frac{\lambda}{d}\right) \cdot \epsilon |V| \geq \frac{2 \lambda}{d} |V|$
by \cref{clm:sc:Si} and the assumption that $\frac{\lambda}{d} \leq \frac{\epsilon}{3}$. 
Using the expander mixing lemma again on $V$ implies
\begin{align}
    e_G(V,V) \leq \frac{d \cdot |V|^2}{|V|} + \lambda \cdot |V| = (d + \lambda) |V|.
\end{align}
Recalling that
$e_G(S,S)$ is twice the number of edges within $G[S]$ for any $S \subseteq V$,
we now estimate the value of $\psi^{(\ttt)}$ as
\begin{align}
    \val_G(\psi^{(\ttt)}) 
    \underbrace{\geq}_{\text{\cref{clm:sc:Lv}}} \frac{e_G(S^{(\ttt)},S^{(\ttt)})}{e_G(V,V)}
    \geq \frac{\epsilon^2 \cdot \left(1-\frac{3 \lambda}{\epsilon d} \right) \cdot d|V|}{(d+\lambda)|V|}
    = \underbrace{\epsilon^2 \cdot \frac{d}{d+\lambda} \cdot
        \left(1-\frac{3 \lambda}{\epsilon d} \right)}_{= \epsilon'}.
\end{align}
Consider finally transforming $\psi^\ini$ into $\psi^\tar$.
Since $\psi^{(\ttt)}$ and $\psi^{(\ttt+1)}$ differ in at most one vertex, it holds that
\begin{align}
    \val_G(\psi^{(\ttt)} \reco \psi^{(\ttt+1)})
    = \min\Bigl\{ \val_G(\psi^{(\ttt)}), \val_G(\psi^{(\ttt+1)}) \Bigr\}.
\end{align}
Therefore, we obtain
\begin{align}
    \val_G(\psi^\ini \reco \psi^\tar)
    \geq \min_{1 \leq \ttt \leq \TTT-1} \val_G(\psi^{(\ttt)} \reco \psi^{(\ttt+1)})
    = \min_{1 \leq \ttt \leq \TTT} \val_G(\psi^{(\ttt)})
    \geq \epsilon',
\end{align}
completing the proof.
\end{proof}

\section{\NP-hardness of Approximation for \MaxminBCSPReconf}
\label{sec:NP}

We show \NP-hardness of approximating \MaxminBCSPReconf
within a factor better than $\frac{3}{4}$.
For every numbers $0 \leq s \leq c \leq 1$,
\prb{Gap$_{c,s}$ \BCSP}
is defined as a decision problem asking
for a constraint graph $G = (V,E,\Sigma,\Pi)$,
whether
$\exists \psi \colon V \to \Sigma, \val_G(\psi) \geq c$ or
$\forall \psi \colon V \to \Sigma, \val_G(\psi) < s$.

\begin{theorem}
\label{thm:NP-BCSPR}
    For every number $\epsilon \in (0,1)$ and
    positive integer $W \geq 2$,
    there exists a gap-preserving reduction from 
    \prb{Gap$_{1,1-\epsilon}$ \BCSP[$_W$]}
    to
    \prb{Gap$_{1,1-\frac{\epsilon}{4}}$ \BCSPReconf[$_{W'}$]},
    where $W' = \left(\frac{W(W+1)}{2}\right)^4$.
    In particular,
    for every number $\epsilon \in (0,1)$,
    there exists a positive integer $W \in \bbN$ such that
    \prb{Gap$_{1,\frac{3}{4}+\epsilon}$ \BCSPReconf[$_W$]} is \NP-hard.
\end{theorem}\noindent
The last statement follows from
the parallel repetition theorem \cite{raz1998parallel} and the PCP theorem~\cite{arora1998probabilistic,arora1998proof}.
Our proof consists of
a gap-preserving reduction
from \prb{Max \BCSP} to \MaxminQCSPReconf (\cref{lem:BCSP-4CSPR})
followed by that
from \MaxminQCSPReconf to \MaxminBCSPReconf (\cref{lem:4CSPR-BCSPR}).

\begin{remark}
    Since the underlying graph is no longer biregular or expander,
    \cref{thm:NP-BCSPR} may not be applied to \cref{thm:sc}.
    Besides, though \cref{thm:NP-BCSPR} along with \cite{ohsaka2023gap}
    implies \NP-hardness of approximation for \MinmaxSetCoverReconf,
    \cite{ohsaka2023gap} entails the degree reduction step and thus would
    result in only a tiny inapproximability factor.
\end{remark}

We first reduce from \prb{Max \BCSP} to \MaxminQCSPReconf
in a gap-preserving manner,
whose proof uses a similar idea of \cite[Theorem 5]{ito2011complexity}.

\begin{lemma}
\label{lem:BCSP-4CSPR}
    For every number $\epsilon \in (0,1)$ and positive integer $W \geq 2$,
    there exists a gap-preserving reduction from 
    \prb{Gap$_{1, 1-\epsilon}$ \BCSP[$_W$]} to
    \prb{Gap$_{1, 1-\epsilon}$ \QCSPReconf[$_W$]}.
\end{lemma}
\begin{proof}
Given an instance $G=(V,E,\Sigma,\Pi)$ of \prb{Max \BCSP[$_W$]}, where $|\Sigma|=W$,
we create an instance $(\psi'^\ini,\psi'^\tar; G')$ of \MaxminQCSPReconf.
The quaternary constraint graph $G'=(V',E',\Sigma,\Pi')$ on the same alphabet $\Sigma$ is defined as follows:
\begin{description}
    \item[\textbf{Underlying graph}:]
        Create a pair of fresh vertices $x,y$ not in $V$.
        Then, $(V',E')$ is defined as follows:
        \begin{align}
            V' & \coloneq V \cup \{x,y\}, \\
            E' & \coloneq \Bigl\{ (v,w,x,y) \Bigm| (v,w) \in E \Bigr\}.
        \end{align}
    \item[\textbf{Constraints}:]
        The constraint $\pi'_{e'} \subseteq \Sigma^{e'}$
        for each hyperedge $e' = (v,w,x,y) \in E'$ is defined as follows:
        \begin{align}
            \pi'_{e'} \coloneq
            \Bigl\{
                (\alpha_v, \alpha_w, \beta, \gamma) \in \Sigma^{e'}
                \Bigm| (\alpha_v, \alpha_w) \in \pi_{(v,w)} \text{ or } \beta = \gamma
            \Bigr\}.
        \end{align}
        In other words, $e'$ is satisfied if
        $\pi_{(v,w)}$ is satisfied \emph{or}
        $x$ and $y$ have the same value.
\end{description}
Construct two assignments $\psi'^\ini, \psi'^\tar \colon V' \to \Sigma$ for $G'$ such that
$\psi'^\ini(v) \coloneq \aaa$ and
$\psi'^\tar(v) \coloneq \bbb$
for all $v \in V'$, where
$\aaa$ and $\bbb$ are any two distinct values in $\Sigma$.
Observe easily that $\psi'^\ini$ and $\psi'^\tar$ satisfy $G'$,
completing the description of the reduction.

We first prove the completeness; i.e.,
$\exists \psi \colon V \to \Sigma, \val_G(\psi) = 1$ implies
$\val_{G'}(\psi'^\ini \reco \psi'^\tar) = 1$.
Let $\psi \colon V \to \Sigma$ be a satisfying assignment for $G$.
Consider a reconfiguration sequence $\sqpsi'$ from $\psi'^\ini$ to $\psi'^\tar$
obtained by the following procedure:

\begin{itembox}[l]{\textbf{Reconfiguration sequence $\sqpsi'$ from $\psi'^\ini$ to $\psi'^\tar$.}}
\begin{algorithmic}[1]
    \For{\textbf{each} $v$ in $V$}
        \State change the current value of $v$ from $\aaa$ to $\psi(v)$.
    \EndFor
    \State change the current values of $x$ and $y$ from $\aaa$ to $\bbb$ individually.
    \For{\textbf{each} $v$ in $V$}
        \State change the current value of $v$ from $\psi(v)$ to $\bbb$.
    \EndFor
\end{algorithmic}
\end{itembox}
In any intermediate step,
each hyperedge $(v,w,x,y) \in E'$ is satisfied since
$\pi_{(v,w)}$ is satisfied or $x$ and $y$ have the same value,
implying that
$\val_{G'}(\psi'^\ini \reco \psi'^\tar) \geq \val_{G'}(\sqpsi') = 1$,
as desired.

We then prove the soundness; i.e.,
$\forall \psi \colon V \to \Sigma, \val_G(\psi) < 1-\epsilon$
implies
$\val_{G'}(\psi'^\ini \reco \psi'^\tar) < 1-\epsilon$.
Let $\sqpsi'$ be a reconfiguration sequence from $\psi'^\ini$ to $\psi'^\tar$ 
whose value is equal to
$\val_{G'}(\psi'^\ini \reco \psi'^\tar)$.
Since
$(\psi'^\ini(x),\psi'^\ini(y)) = (\aaa,\aaa)$ and
$(\psi'^\tar(x),\psi'^\tar(y)) = (\bbb,\bbb)$,
$\sqpsi'$ must include an assignment $\psi'^\circ$ such that
$\psi'^\circ(x) \neq \psi'^\circ(y)$.
Calculating the value of such $\psi'^\circ$, we derive
\begin{align}
\begin{aligned}
    \val_{G'}(\psi'^\circ)
    & = \frac{1}{|E'|}
        \sum_{(v,w,x,y) \in E'} \Bigl\llbracket
            (\psi'^\circ(v), \psi'^\circ(w)) \in \pi_{(v,w)} \text{ or }
            \underbrace{\psi'^\circ(x) = \psi'^\circ(y)}_{\text{not true}}
        \Bigr\rrbracket \\
    & = \frac{1}{|E|} \sum_{(v,w) \in E} \Bigl\llbracket
            (\psi'^\circ(v), \psi'^\circ(w)) \in \pi_{(v,w)}
        \Bigr\rrbracket \\
    & = \val_G(\psi'^\circ|_V),
\end{aligned}
\end{align}
where $\psi'^\circ|_V$ is the restriction of $\psi'^\circ$ onto $V$; thus,
$\val_{G'}(\psi'^\circ) = \val_G(\psi'^\circ|_V) < 1-\epsilon$.
Consequently, we obtain
$\val_{G'}(\psi'^\ini \reco \psi'^\tar) = \val_{G'}(\sqpsi') < 1-\epsilon$,
as desired.
\end{proof}

We then present a gap-preserving reduction from
\MaxminQCSPReconf to \MaxminBCSPReconf.
Indeed, we demonstrate a general one described below,
whose proof is based on a reduction by \citet{fortnow1994power}
(see also \cite[Proposition~6.2]{radhakrishnan2007dinurs}).

\begin{lemma}
\label{lem:4CSPR-BCSPR}
    For every number $\epsilon \in (0,1)$,
    positive integers $q \geq 3$ and $W \geq 2$,
    there exists a gap-preserving reduction from
    \prb{Gap$_{1,1-\epsilon}$ $q$-CSP$_W$ Reconfiguration}
    to
    \prb{Gap$_{1,1-\frac{\epsilon}{q}}$ \BCSPReconf[$_{W'}$]},
    where $W' = \left(\frac{W(W+1)}{2}\right)^q$.
\end{lemma}
\begin{proof}
We describe a gap-preserving reduction from
\prb{Maxmin $q$-CSP$_W$ Reconfiguration}
to
\MaxminBCSPReconf[$_{W'}$].
Let $(\psi^\ini, \psi^\tar; G)$ be an instance of \prb{Maxmin $q$-CSP$_W$ Reconfiguration},
where $G = (V,E,\Sigma,\Pi = (\pi_e)_{e \in E})$ is a $q$-ary constraint graph such that
$|\Sigma| = W$, and
$\psi^\ini$ and $\psi^\tar$ satisfy $G$.
Below, we construct an instance $(\psi'^\ini, \psi'^\tar; G')$ of
\MaxminBCSPReconf[$_{W'}$].
The constraint graph $G'=(V',E',\Sigma',\Pi')$ is defined as follows:

\begin{description}
    \item[\textbf{Underlying graph:}]
        $(V',E')$ is a bipartite graph representation of $(V,E)$;
        i.e., $V'$ is made up of a bipartition $(V, E)$, where
        each vertex of $V'$ corresponds to a vertex or hyperedge of $G$, and
        $E'$ includes an edge $(v,e)$ between $v \in V$ and $e \in E$ if $v$ appears in hyperedge $e$.
        Note that $|E'| = q|E|$.
    
    \item[\textbf{Alphabet:}]
        Use the alphabet squaring trick \cite{ohsaka2023gap} to define
        \begin{align}
            \Sigma' \coloneq
            \left(
                \Bigl\{ \{\alpha\} \Bigm| \alpha \in \Sigma \Bigr\}
                \cup \Bigl\{ \{\alpha, \beta\} \Bigm| \alpha \neq \beta \in \Sigma \Bigr\}
            \right)^q.
        \end{align}
        Note that
        $|\Sigma'| = \left(\frac{W(W+1)}{2}\right)^q$.
        Given an assignment $\psi' \colon V' \to \Sigma'$ for $G'$,
        $\psi'(v)[i]$ designates the \nth{$i$} coordinate of $\psi'(v)$
        for a vertex $v \in V$, and
        $\psi'(e)[v_i]$ designates the \nth{$i$} coordinate of $\psi'(e)$
        for a hyperedge $e = (v_1,\ldots,v_r) \in E$.
        Briefly, for a hyperedge $e$ including a vertex $v$,
        both $\psi'(v)[1]$ and $\psi'(e)[v]$ are
        supposed to claim the value of $v$;
        we will not be aware of the values of
        $\psi'(v)[2], \ldots, \psi'(v)[q]$.
    
    \item[\textbf{Constraints:}]
        The constraint $\pi'_{e'} \subseteq \Sigma'^{e'}$
        for each edge $e' = (v_i, e) \in E'$ with $e = (v_1,\ldots,v_q) \in E$
        is defined as follows:
        \begin{align}
        \label{eq:4CSPR-BCSPR:pi}
            \pi'_{e'} \coloneq
            \left\{
                (\alpha_{v_i}, \alpha_e) \in \Sigma'^{e'}
                \;\middle|\;
                \alpha_{v_i}[1] \subseteq \alpha_e[v_i] \text{ and } 
                \bigtimes_{j \in [q]} \alpha_e[v_j] \subseteq \pi_e
            \right\}
        \end{align}
\end{description}
Construct any two assignments $\psi'^\ini, \psi'^\tar \colon V' \to \Sigma'$ for $G'$ such that
the following hold:
\begin{itemize}
    \item
        $\psi'^\ini(v)[1] \coloneq \{\psi^\ini(v)\}$
        and
        $\psi'^\tar(v)[1] \coloneq \{\psi^\tar(v)\}$
        for all $v \in V$.
    \item
        The remaining entries 
        $\psi'^\ini(v)[2], \ldots, \psi'^\ini(v)[q]$ and
        $\psi'^\tar(v)[2], \ldots, \psi'^\tar(v)[q]$
        for all $v \in V$ are arbitrary.
    \item
        $\psi'^\ini(e)[v_i] \coloneq \{\psi^\ini(v_i)\}$ and
        $\psi'^\tar(e)[v_i] \coloneq \{\psi^\tar(v_i)\}$
        for all $e=(v_1, \ldots, v_q) \in E$ and $i \in [q]$.
\end{itemize}
Observe that $\psi'^\ini$ and $\psi'^\tar$ satisfy $G'$, thereby completing the description of the reduction.

We first prove the completeness; i.e.,
$\val_G(\psi^\ini \reco \psi^\tar) = 1$ implies $\val_{G'}(\psi'^\ini \reco \psi'^\tar) = 1$.
It suffices to consider the case that $\psi^\ini$ and $\psi^\tar$ differ in exactly one vertex, say, $v \in V$.
Define $\alpha \coloneq \psi^\ini(v)$ and $\beta \coloneq \psi^\tar(v)$.
By construction, it holds that
$\psi'^\ini(v)[1] = \{\alpha\}$,
$\psi'^\tar(v)[1] = \{\beta\}$,
$\psi'^\ini(e)[v] = \{\alpha\}$, and
$\psi'^\tar(e)[v] = \{\beta\}$
for every hyperedge $e$ including $v$.
Consider a reconfiguration sequence $\sqpsi'$ from $\psi'^\ini$ to $\psi'^\tar$ obtained by the following procedure:

\begin{itembox}[l]{\textbf{Reconfiguration sequence $\sqpsi'$ from $\psi'^\ini$ to $\psi'^\tar$.}}
\begin{algorithmic}[1]
\For{\textbf{each} hyperedge $e$ including $v$}
    \State change the \nth{$v$} coordinate of $e$'s value from $\{\alpha\}$ to $\{\alpha,\beta\}$.
\EndFor
\State change the first coordinate of the value of $v$ from $\{\alpha\}$ to $\{\beta\}$.
\For{\textbf{each} hyperedge $e$ including $v$}
    \State change the \nth{$v$} coordinate of $e$'s value from $\{\alpha,\beta\}$ to $\{\beta\}$.
\EndFor
\end{algorithmic}
\end{itembox}
In any intermediate step of this transformation,
every constraint is satisfied; i.e.,
$\val_G(\psi'^\ini \reco \psi'^\tar) \geq \val_{G'}(\sqpsi') = 1$, as desired.

We then prove the soundness; i.e.,
$\val_G(\psi^\ini \reco \psi^\tar) < 1-\epsilon$
implies
$\val_{G'}(\psi'^\ini \reco \psi'^\tar) < 1-\frac{\epsilon}{q}$.
Let $\sqpsi' = ( \psi'^{(1)}, \ldots, \psi'^{(\TTT)} )$
be a reconfiguration sequence from $\psi'^\ini$ to $\psi'^\tar$ whose value is equal to
$\val_{G'}(\psi'^\ini \reco \psi'^\tar)$.
Construct then a sequence
$\sqpsi = (\psi^{(1)}, \ldots, \psi^{(\TTT)})$
such that
each assignment $\psi^{(\ttt)} \colon V \to \Sigma$ is determined based on
the first coordinate of $\psi'^{(\ttt)}$; i.e.,
it holds that $\psi^{(\ttt)}(v) \in \psi'^{(\ttt)}(v)[1]$
for all $v\in V$.\footnote{
    If $\psi'^{(\ttt)}(v)[1]$ consists of a pair of values of $\Sigma$,
    either of them is selected according to any prefixed order over $\Sigma$.
}
Observing that $\sqpsi$ is a valid reconfiguration sequence from $\psi^\ini$ to $\psi^\tar$,
we have
$\val_G(\sqpsi) < 1-\epsilon$;
in particular, there exists some $\psi^{(\ttt)}$ such that
$\val_G(\psi^{(\ttt)}) < 1-\epsilon$.
Suppose $\psi^{(\ttt)} $ does not satisfy a hyperedge $e = (v_1,\ldots, v_q) \in E$.
Then, we claim that $\psi'^{(\ttt)}$ must violate \emph{at least one edge} incident to $e$ in $G'$
because otherwise, by \cref{eq:4CSPR-BCSPR:pi}, we obtain 
\begin{align}
\begin{aligned}
    \Bigl(\psi^{(\ttt)}(v_1), \ldots, \psi^{(\ttt)}(v_q)\Bigr)
    & \in \psi'^{(\ttt)}(v_1)[1] \times \cdots \times \psi'^{(\ttt)}(v_q)[1] \\
    & \subseteq \psi'^{(\ttt)}(e)[v_1] \times \cdots \times \psi'^{(\ttt)}(e)[v_q] \\
    & \subseteq \pi_e,
\end{aligned}
\end{align}
i.e., $\psi^{(\ttt)}$ satisfies $e$, a contradiction.
Thus, $\psi'^{(\ttt)}$ violates more than $\epsilon|E|$ edges of $G'$ in total.
Consequently, we derive
\begin{align}
    \val_{G'}(\psi'^\ini \reco \psi'^\tar)
    = \val_{G'}(\sqpsi')
    \leq \val_{G'}(\psi'^{(\ttt)})
    < \frac{|E'|-\epsilon|E|}{|E'|}
    \underbrace{=}_{|E'| = q|E|} 1-\frac{\epsilon}{q},
\end{align}
completing the soundness.
\end{proof}

\paragraph{Acknowledgments.}
I wish to thank the anonymous referees for their careful reading and pointing out
the reference \cite{bogdanov2005gap} and the overlap gap property, e.g., \cite{gamarnik2021overlap,achlioptas2011solution,mezard2005clustering,gamarnik2017limits,wein2021optimal}.

\printbibliography[heading=bibintoc]

\appendix
\section{Omitted Proof}
\label{app:proof}

\begin{proof}[Proof of \cref{lem:lambda}]
By definition of $\lambda(\cdot)$, we have
\begin{align}
\begin{aligned}
    \lambda(G \uplus H)
    & = \max_{\vec{x} \neq \vec{0}, \vec{x} \perp \vec{1}}
    \frac{\|(\mat{A}_G + \mat{A}_H) \vec{x}\|_2}{\|\vec{x}\|_2} \\
    & \leq \max_{\vec{x} \neq \vec{0}, \vec{x} \perp \vec{1}}
    \frac{\|\mat{A}_G \vec{x}\|_2}{\|\vec{x}\|_2}
    + \max_{\vec{y} \neq \vec{0}, \vec{y} \perp \vec{1}}
    \frac{\|\mat{A}_H \vec{y}\|_2}{\|\vec{y}\|_2} \\
    & = \lambda(G) + \lambda(H).
\end{aligned}
\end{align}
Recall that
\begin{align}
    \lambda(GH)
    = \max_{\vec{x} \neq \vec{0}, \vec{x} \perp \vec{1}}
    \frac{\|\mat{A}_G\mat{A}_H \vec{x}\|_2}{\|\vec{x}\|_2}.
\end{align}
Since $\mat{A}_H \vec{x}$ is orthogonal to $\vec{1}$ whenever so is $\vec{x}$, we have
\begin{align}
\begin{aligned}
    \lambda(GH)
    & = \max_{\vec{x} \neq \vec{0}, \vec{x} \perp \vec{1}}
    \frac{\|\mat{A}_G (\mat{A}_H \vec{x})\|_2}{\|\mat{A}_H \vec{x}\|_2}
    \cdot \frac{\|\mat{A}_H\vec{x}\|_2}{\|\vec{x}\|_2} \\
    & \leq \max_{\vec{y} \neq \vec{0}, \vec{y} \perp \vec{1}}
    \frac{\|\mat{A}_G \vec{y}\|_2}{\|\vec{y}\|_2}
    \cdot \max_{\vec{x} \neq \vec{0}, \vec{x} \perp \vec{1}}
    \frac{\|\mat{A}_H\vec{x}\|_2}{\|\vec{x}\|_2} \\
    & = \lambda(G) \cdot \lambda(H),
\end{aligned}
\end{align}
finishing the proof.
\end{proof}

\end{document}